\documentclass[11pt]{article}

\usepackage[english]{babel}
\usepackage[utf8]{inputenc}
\usepackage[T1]{fontenc}
\usepackage[margin=1in]{geometry}

\usepackage{microtype}
\usepackage{lmodern}

\usepackage{amsmath,amssymb,amsthm,mathtools,mathrsfs}
\usepackage{bm}
\usepackage{bbm}

\usepackage{graphicx}
\usepackage{subcaption}
\usepackage{booktabs}
\usepackage{multirow}
\usepackage{tabularx}
\usepackage{float}
\usepackage{rotating}

\usepackage[round]{natbib}

\usepackage[linesnumbered,ruled,vlined]{algorithm2e}
\SetKwInput{KwInput}{Input}
\SetKwInput{KwOutput}{Output}
\SetKwInput{KwInit}{Initialization}

\usepackage[hidelinks]{hyperref}
\usepackage{url}

\usepackage[normalem]{ulem} 
\usepackage{xcolor}
\usepackage{nicefrac}
\usepackage{lipsum}

\definecolor{light-gray}{gray}{0.6}
\definecolor{darkgreen}{rgb}{0.1, 0.45, 0.15}
\definecolor{darkblue}{rgb}{0.1, 0.15, 0.45}

\def\boxit#1{\vbox{\hrule\hbox{\vrule\kern6pt
          \vbox{\kern6pt#1\kern6pt}\kern6pt\vrule}\hrule}}

\newcommand{\R}{\mathbb{R}}

\newcommand{\ee}{\mathbb{E}}

  \newcommand{\cD}{\mathcal{D}}

\newcommand{\cJ}{\mathcal{J}}  
\newcommand{\cL}{\mathcal{L}}  \newcommand{\cM}{\mathcal{M}}
  
\newcommand{\cS}{\mathcal{S}}  
  
\newcommand{\cT}{\mathcal{T}}

\newcommand{\reals}{\mathbb{R}}
\renewcommand{\v}[1]{\boldsymbol{#1}}

\newcommand{\lt}{\left}
\newcommand{\rt}{\right}
\newcommand{\wt}{\widetilde}
\newcommand{\wh}{\widehat}

\newcommand{\conv}{\operatorname{conv}}
\newcommand{\ind}{\mathbb{I}}

\newcommand{\argmin}{\operatornamewithlimits{argmin}}

\DeclareMathOperator*{\minimize}{min\,}

\DeclareMathOperator*{\subjectto}{subject\,\, to\,}
\DeclareMathOperator*{\diag}{Diag\,}

\newcommand{\bbt}{\boldsymbol{\beta}}

\newcommand{\hbt}{\widehat{\bbt}}
\newcommand{\tbt}{\widetilde{\bbt}}

\newcolumntype{C}{>{\centering\arraybackslash}X}

\newtheorem{theorem}{Theorem}

\newtheorem{proposition}{Proposition}

\newtheorem{remark}{Remark}

\usepackage{authblk}

\title{Parsimonious Subset Selection for Generalized Linear Models with Biomedical Applications}

\author[1]{Anant Mathur\thanks{Corresponding author: \texttt{anant.mathur@unsw.edu.au}}}
\author[2,3]{Benoit Liquet}
\author[2]{Samuel Muller}
\author[1]{Sarat Moka}

\affil[1]{School of Mathematics and Statistics, University of New South Wales, NSW, Australia}
\affil[2]{School of Mathematical and Physical Sciences, Macquarie University, NSW, Australia}
\affil[3]{Laboratoire de Math\'ematiques et de leurs Applications, Universit\'e de Pau et des Pays de l'Adour, Pau, France}

\date{} 

\begin{document}

\maketitle

\begin{abstract}
High-dimensional biomedical studies require models that are
simultaneously accurate, sparse, and interpretable, yet exact best
subset selection for generalized linear models is computationally
intractable. We develop a scalable method that combines a continuous
Boolean relaxation of the subset problem with a Frank--Wolfe
algorithm driven by envelope gradients. The resulting method, which we refer to as COMBSS-GLM, is simple
to implement, requires one penalized generalized linear model fit per
iteration, and produces sparse models along a model-size path.
Theoretically, we identify a curvature-based parameter regime in which
the relaxed objective is concave in the selection weights, implying
that global minimizers occur at binary corners. Empirically, in
logistic and multinomial simulations across low- and high-dimensional
correlated settings, the proposed method consistently improves
variable-selection quality relative to established penalised likelihood
competitors while maintaining strong predictive performance. In biomedical applications, it recovers established loci in a binary-outcome rice genome-wide association study and achieves perfect multiclass test accuracy on the Khan SRBCT cancer dataset using a small subset of genes. Open-source implementations are available in R at \url{https://github.com/benoit-liquet/COMBSS-GLM-R} and in Python at \url{https://github.com/saratmoka/COMBSS-GLM-Python}.
\end{abstract}


\section{Introduction}
\label{s:intro}

Generalized linear models (GLMs), introduced by Nelder and Wedderburn
\citep{nelder1972generalized}, extend classical linear regression to
handle non-Gaussian outcomes, such as binary indicators, counts
(e.g., Poisson), and positive-valued responses (e.g., exponential),
by modeling the response through an exponential-family distribution
and linking its conditional mean to a linear predictor via an
appropriate link function \citep{neuhaus2011generalized,
hastie2017generalized}. This framework unifies many widely used
models, including linear, logistic, and Poisson regression. Common
link functions include the identity link (Gaussian), logit link
(Binomial), log link (Poisson), and reciprocal link (exponential).

Recent advances in data acquisition and storage have enabled the
collection of complex, high-dimensional datasets across many domains,
including biology, engineering, economics, finance, and the health
sciences \citep{FanLi2006, fan2014challenges}.
In biomedical and clinical research in particular, modern studies
routinely record thousands to millions of candidate predictors, for example, single nucleotide polymorphism (SNP) genotypes in genome-wide association studies \citep{mccouch2016open}, gene expression measurements in cancer profiling \citep{khan2001classification}, imaging-derived variables, and electronic
health records. In such settings, \emph{variable selection} is essential for both predictive
performance and scientific interpretability \citep{guyon2003introduction}.
By identifying a parsimonious set of relevant predictors, feature
selection mitigates overfitting and improves generalization, sharpens
estimation and inference by removing irrelevant or redundant
variables, and yields models that are easier to interpret and deploy
in practice \citep{hastie2015statistical}. 

A classical approach to variable selection is \emph{best subset
selection}: for a prescribed model size $k$, one seeks the subset of
$k$ predictors that yields the best fit under a chosen criterion
(typically the negative log-likelihood, possibly with a stabilizing
quadratic shrinkage term) \citep{beale1967discarding}. 
Unlike unconstrained GLM estimation, this
requires searching over an exponentially large collection of candidate
subsets, which makes the problem computationally challenging; in fact,
it is NP-hard even in the special case of linear regression
\citep{Natarajan1995}. In practice, the subset size $k$ (and any
accompanying shrinkage level) is treated as a tuning choice and is
selected to optimize predictive performance on unseen data \citep{hastie2009elements}. 

Addressing this computational challenge for generalized linear models is the focus of the present paper. Building on the continuous Boolean relaxation approach of \citet{moka2024combss}, which was developed for linear regression, and subsequently adapted to best
subset selection in linear dimension reduction models \citep{BSS_PCA_PLS}, we establish a general method for best subset selection in the GLM setting.

The key idea is to relax the intractable combinatorial search over binary subsets to a smooth optimization over the polytope $\{\v t \in [0, 1]^p : \sum_{j = 1}^p t_j = k\}$ for a prescribed model size $k$, where $t_j \in \{0,1\}$ encodes variable exclusion or inclusion. We solve this constrained problem using a variant of the Frank–Wolfe algorithm \citep{frank1956algorithm,jaggi2013revisiting}. In contrast to the original COMBSS method for linear regression, which replaces the constraint with a penalty and yields an unconstrained formulation, our approach enforces exact cardinality, thereby enabling systematic exploration of models at each prescribed $k$. We further equip the Frank--Wolfe algorithm with a homotopy scheme that progressively deforms the relaxed landscape until its concave, ensuring that the final solution converges to a binary corner of the polytope.
The resulting algorithm, which we refer to as COMBSS-GLM, is simple to implement, leverages standard
GLM solvers, such as \texttt{glmnet} \citep{friedman2010regularization}, as black-box inner routines, and scales to
high-dimensional problems with thousands of predictors. 

 Although COMBSS-GLM applies to any GLM with a convex negative log-likelihood, and the theoretical results are established at this level of generality, we conduct experiments primarily for logistic and multinomial regression, which are the most prevalent GLMs in high-dimensional biomedical classification. Our key contributions are as follows.

\begin{itemize}
  \item We formulate best subset selection in GLMs through a
    continuous Boolean relaxation that preserves the fixed-model-size
    structure and supports direct optimization on a simplex slice.
  \item We develop a homotopy Frank--Wolfe algorithm with envelope
    gradients; each iteration requires one penalized GLM fit and a
    $k$-smallest-component vertex update, with no projection step.
  \item We establish a curvature-based large-parameter regime under
    which the relaxed objective is concave in the selection weights,
    implying that global minimizers lie at binary corners; for
    logistic and baseline multinomial models, we provide explicit
    sufficient thresholds.
  \item Through simulations and two biomedical case studies ,a high-dimensional rice GWAS (p = 158{,}210 SNPs; \citep{mccouch2016open}) and a multiclass cancer gene-expression problem (p = 2{,}308, C=4 tumour types; \citep{khan2001classification}), we show that COMBSS-GLM improves sparse recovery while maintaining strong predictive performance, and yields highly parsimonious models in practice.
\end{itemize}

The remaining paper is organized as follows. In Section~\ref{s:methodology}, we introduce relevant notation, formulate best subset selection for GLMs, and develop the continuous Boolean relaxation underlying COMBSS-GLM.
In Section~\ref{sec:algorithm}, we present the COMBSS-GLM algorithm, establish the concavity threshold and monotonicity properties of the
relaxed objective, and discuss hyperparameter tuning including the default choice of the homotopy grid size, $N$.
In Section~\ref{s:simulations}, we report simulation studies comparing COMBSS-GLM against penalized regression competitors in the logistic
regression setting.
In Section~\ref{sec:applications}, we apply the COMBSS-GLM algorithm to two biomedical
datasets: a rice genome-wide association study (GWAS) study and the small round blue cell tumor (SRBCT) Khan cancer gene-expression
dataset. Concluding remarks are discussed in Section~\ref{sec:conclusion}.
Proofs and additional numerical results are collated in the
supplementary material.

\section{Methodology}
\label{s:methodology}


Throughout the paper, all vectors are understood to be column vectors.
For a scalar $a_0$ and a vector $\v a = (a_1, a_2, \dots, a_p)^{\top}$, we
write $(a_0, \v a)$ to denote the concatenated vector
$(a_0, a_1, a_2, \dots, a_p)^{\top}$.
For a binary vector $\v s = (s_1, \dots, s_p)^{\top} \in \{0,1\}^{p}$, we
denote by $|\v s|$ the number of nonzero (equal to one) entries in
$\v s$.
For vectors $\v a, \v b \in \reals^p$, the Hadamard (elementwise)
product is denoted by $\v a \odot \v b := (a_1 b_1, \dots, a_p b_p)^{\top}$.
For a vector $\v v \in \reals^p$, $\diag(\v v)$ denotes the
$p \times p$ diagonal matrix with diagonal entries given by $\v v$.
We denote by $\v 1 \in \reals^p$ the vector of all ones.
For a function~$f$, $\argmin_x f(x)$ denotes the set of minimizers of
$f$ over its domain.
All vector square roots are understood elementwise. For a set $A \subseteq \reals^p$, $\conv\{A\}$ denotes the convex hull of
$A$, namely the set of all convex combinations of points in $A$.

\subsection{Preliminaries}
\label{ss:prelim}

We now provide a mathematical formulation of GLMs and best subset selection for these models.
Let $\cD=(X,\v y)$, where $X\in\R^{n\times(p+1)}$ is a design matrix
whose rows have the form $(1,x_1,\dots,x_p)$ (the leading $1$
corresponding to the intercept), and $\v y=(y_1,\dots,y_n)^\top$ is
the response vector with entries drawn from an exponential-family
distribution.
In a GLM, an invertible link function $g$ relates the conditional mean
to a linear predictor via
\(
  g\!\left(\ee[Y\mid \v x]\right)=\beta_0+\v x^\top\beta,
\)
where $\beta_0\in\R$ is an intercept and
$\v \beta=(\beta_1,\dots,\beta_p)^\top\in\R^p$ is a coefficient vector
\citep{mccullagh2019generalized}.


The parameters $(\beta_0,\bbt)$ are estimated by maximum likelihood.
Writing $\ell(\beta_0,\bbt;\cD):=\log L(\beta_0,\bbt;\cD)$ for the
log-likelihood, the MLE solves
\begin{align}
  \label{eqn:mle-opt}
  \min_{\beta_0\in\reals,\;\bbt\in\reals^p} \; -\ell(\beta_0,\bbt;\cD).
\end{align}
For standard GLMs (Gaussian, Binomial, Poisson),
$-\ell(\beta_0,\bbt;\cD)$ is convex in $(\beta_0,\bbt)$
\citep{mccullagh2019generalized}. However, except in the case of linear regression the
minimizer has no closed form and is computed numerically
\citep{boyd2004convex}.


In this GLM setting, a common formulation of best subset selection is the cardinality-constrained regularized MLE               
\begin{align}
\begin{aligned}
  \minimize_{\beta_0 \in \reals,\; \bbt \in \reals^{p}} \quad
  &-\frac{1}{n}\ell(\beta_0,\bbt;\cD) + \lambda \|\bbt\|_2^2 \\
  \subjectto \quad &\|\bbt\|_0 = k,
\end{aligned}
\label{eqn:bss-l0}
\end{align}
for a prescribed model size $k$, where
$\|\bbt\|_0 := \sum_{j \in \cJ} \ind(\beta_j \neq 0)$ counts the
number of nonzero coefficients subjected to selection (with $\cJ$
typically excluding the intercept and any mandatory variables), and
$\lambda \ge 0$ is an optional ridge penalty.
In practice, $k$ and $\lambda$ are chosen to minimise an out-of-sample
criterion, such as validation or cross-validation error.

\subsection{Boolean Relaxation}
\label{eqn:continuous-relax}

In this section, we provide a continuous relaxation of the best subset
selection problem stated in~\eqref{eqn:bss-l0}.
Towards this, we first reformulate \eqref{eqn:bss-l0} as an equivalent
binary constrained optimization.
Since the design matrix $X$ has $p+1$ columns (including the
intercept), we number them using the index $j = 0, 1,2, \dots, p$.
To keep it general, we assume that $m$ of the variables, in addition to
the intercept, are mandatory and are always included in the model.
Without loss of generality, we assume they are the first $m$ variables
indexed from $1$ to $m$, while the index $0$ corresponds to the
intercept.

For any binary vector $\v s = (s_1, \dots, s_{p-m})^{\top} \in
\{0,1\}^{p-m}$, let $X_{[ \v s]} \in \reals^{n \times (m + 1 + |\v
s|)}$ denote the submatrix of $X$ obtained by retaining the intercept
column, mandatory columns and all other columns with indices $m+j$
where $s_j = 1$.
With this notation, define $\cD_{[\v s]} = (X_{[\v s]}, \v y)$, and
for every $\v s \in \{0,1\}^{p-m}$, let
$(\wh{\beta}_{\v s, 0}, \hbt_{\v s})$ be the low-dimensional MLE
solution of the problem
\begin{align}
  \min_{\beta_0 \in \reals, \, \bbt \in \reals^{|\v s|}}
  - \frac{1}{n}\ell(\beta_0, \bbt; \cD_{[\v s]}) +
  \lambda \|\bbt\|^2_2,
  \label{eqn:llhf_min}
\end{align}
where $\ell(\beta_0, \bbt; \cD_{[\v s]})$ is the log-likelihood
function of the GLM of interest with the dataset $\cD_{[\v s]}$.
Furthermore, define
\[
  \cS_k = \{\v s \in \{0, 1\}^{p-m} : |\v s| = k\},
\]
which is the set of all binary vectors of length $p-m$ with exactly
$k$ ones.
Then, solving \eqref{eqn:bss-l0} is equivalent to solving the binary
constrained problem
\begin{align}
  \begin{aligned}
    \minimize_{{\v s} \in \cS_k} \,\,
    -\frac{1}{n}\ell(\wh{\beta}_{\v s, 0}, \hbt_{\v s};
    \cD_{[\v s]}) + \lambda \|\hbt_{\v s}\|^2_2.
  \end{aligned}
  \label{eqn:dbss}
\end{align}

For any fixed $\v s$, finding a good approximation to the MLE
$(\wh{\beta}_{\v s, 0}, \hbt_{\v s})$ is relatively easy using a
standard software package such as {\sf glmnet} in {\sf R}.
In linear regression, we even have a closed-form expression, namely
the ridge estimator.
However, as mentioned earlier, solving the combinatorial problem
\eqref{eqn:dbss} exactly is known to be NP-hard, making best subset
selection a difficult task for larger values of $p$. For instance, the R package \texttt{bestglm}  \citep{bestglmR} based on the leaps-and-bounds method \citep{furnival2000regressions} can handle only data with dimension $p \leq 30$.

To overcome this challenge, we provide a Boolean relaxation
of~\eqref{eqn:dbss} by moving the constraints from the binary space
$\{0, 1\}^{p-m}$ to the hypercube $[0, 1]^{p-m}$.
For $\v t = (t_1, \dots, t_{p-m})^\top \in [0, 1]^{p-m}$, let
\[
  T_{\v t} = \diag(\underbrace{1, 1, \dots, 1}_{m\;\text{times}},
  t_1, t_2, \dots, t_{p-m})
  \quad\text{and}\quad
  \Gamma_{\v t} = \sqrt{I - T_{\v t}^2},
\]
where $I$ denotes an identity matrix of appropriate dimension.
Note that $\Gamma_{\v t}$ is a square diagonal matrix of size
$p \times p$ with zeros in the first $m$ locations and
$\sqrt{1 - t_j^2}$ for the next $m+j$ locations.
Now take $\delta > 0$ and define
\begin{align}
  h_{\delta, \lambda}(\v t, \beta_0, \bbt) :=
  -\frac{1}{n}\ell(\beta_0, T_{\v t} \bbt; \cD)
  + \lambda \|\bbt\|_2^2
  + \delta \|\Gamma_{\v t}\bbt\|_2^2,
  \label{eqn:h-def}
\end{align}
where $\ell$ is the log-likelihood function of the GLM under
consideration.
Note that $\ell(\beta_0, T_{\v t}\bbt; \cD)$ is obtained by replacing
the coefficient vector $\bbt$ with
\[
  T_{\v t} \bbt =
  (\beta_1, \dots, \beta_m, t_1 \beta_{m+1}, \dots, t_{p-m} \beta_p)^{\top}
\]
in the log-likelihood function~$\ell$.
For each $\v t \in [0, 1]^{p-m}$, let
\begin{align}
  (\wt{\beta}_{\v t, 0}, \tbt_{\v t})
  \in \argmin_{\beta_0, \bbt} \;
  h_{\delta,\lambda}(\v t,\beta_0, \bbt).
  \label{eqn:mle-solution}
\end{align}
Since $\delta>0$, the additional ridge penalty
$\delta \|\Gamma_{\v t}\bbt\|_2^2
= \delta \sum_{j=1}^{p-m} (1-t_j^2)\beta_{m+j}^2$
in $h_{\delta,\lambda}(\v t,\beta_0, \bbt)$ penalizes each coefficient
$\beta_{m+j}$ separately so that we enforce the support implied by
$\v t$: if $t_j=0$ then $(\tbt_{\v t})_{m+j}=0$, whereas when
$t_j=1$ the coefficient $(\tbt_{\v t})_{m+j}$ incurs no penalization.
By defining
\[
  \cT_k := \{\v t \in [0, 1]^{p-m} : \v 1^\top \v t = k\},
\]
which is a relaxation of $\cS_k$, and by using a solution
$(\wt{\beta}_{\v t, 0}, \tbt_{\v t})$ of \eqref{eqn:mle-solution},
we define
\begin{align}
  f_{\delta, \lambda}(\v t) =
  h_{\delta, \lambda} (\v t, \wt{\beta}_{\v t, 0}, \tbt_{\v t}),
  \label{eqn:f_delta}
\end{align}
and write a Boolean relaxation of \eqref{eqn:dbss} as
\begin{align}
  \minimize_{{\v t} \in \cT_{k}} \, f_{\delta, \lambda}(\v t).
  \label{eqn:cbss}
\end{align}
Here, $f_{\delta, \lambda}(\v t)$ depends on $\delta$ through
$(\wt{\beta}_{\v t, 0}, \tbt_{\v t})$.

In the sequel, we develop a gradient-based algorithm for
\eqref{eqn:cbss} that iteratively updates $\v t$ using the gradient
$\nabla f_{\delta,\lambda}(\v t)$ where $\v t \in (0, 1)^{p-m}$.
In the special case of linear regression, the inner minimizer
$(\wt{\beta}_{\v t,0},\tbt_{\v t})$ admits a closed-form expression
for interior points $\v t\in(0,1)^{p-m}$, and one may compute
$\nabla f_{\delta,\lambda}(\v t)$ by differentiating this closed form
with respect to $\v t$; see \citep{moka2024combss}.
For general GLMs, however, $(\wt{\beta}_{\v t,0},\tbt_{\v t})$
typically has no closed-form expression, and differentiating the
optimizer is undesirable.
Fortunately, Danskin's envelope theorem yields an exact expression for
$\nabla f_{\delta,\lambda}(\v t)$ on the interior that does not
require differentiating $(\wt{\beta}_{\v t,0},\tbt_{\v t})$; see
Section~\ref{sec:inner-solvers} for details.

\begin{figure}[t]
  \centering
  \includegraphics[width=0.5\textwidth]{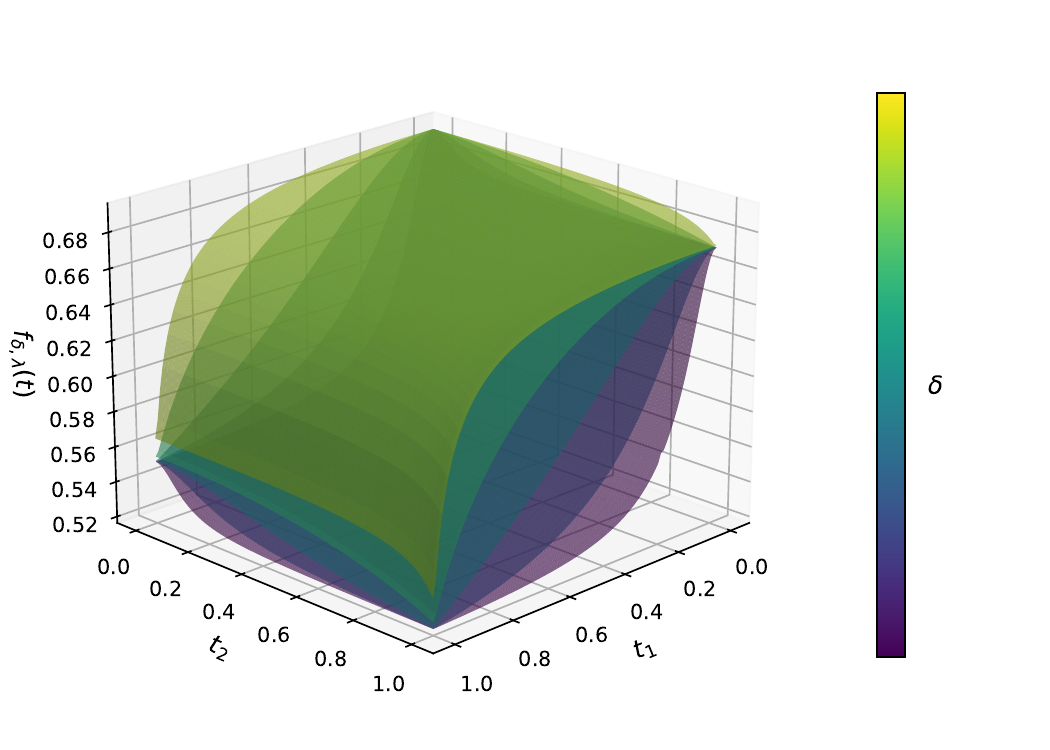}
  \caption{Boolean relaxation surface $f_{\delta,\lambda}(\v t)$ for
    logistic regression with $\lambda = 0$ and varying curvature
    parameter $\delta$. The surfaces are plotted over the domain
    $\v t = (t_1, t_2) \in [0,1]^2$ with five geometrically
    increasing values of $\delta$ shown in different colours according
    to the colour bar. As $\delta$ increases, the relaxed objective
    becomes more peaked, driving the solution toward binary corners.}
  \label{fig:logistic-surfaces}
\end{figure}

\section{Algorithm}
\label{sec:algorithm}

We now present a simple homotopy algorithm, adapted from
\citet{jaggi2013revisiting} and \citet{moka2025scalable}, for
minimizing the relaxed objective in \eqref{eqn:cbss}.
The method increases the curvature parameter $\delta$ geometrically and
performs a Frank--Wolfe (FW) update on $\v t$ at each stage using an
\emph{envelope (sub)gradient} of the value function
\[
  f_{\delta,\lambda}(\v t)
  := \min_{\beta_0,\bbt}\,h_{\delta,\lambda}(\v t,\beta_0,\bbt).
\]
Let
$\cM_{\delta,\lambda}(\v t)
:=\argmin_{\beta_0,\bbt}\,h_{\delta,\lambda}(\v t,\beta_0,\bbt)$,
and define the envelope-gradient set
\[
  \partial^{\mathrm{env}} f_{\delta,\lambda}(\v t)
  :=
  \operatorname{conv}\Bigl\{
  \nabla_{\v t}h_{\delta,\lambda}(\v t,\beta_0,\bbt)
  :(\beta_0,\bbt)\in \cM_{\delta,\lambda}(\v t)
  \Bigr\}.
\]
Given any $(\wt{\beta}_{\v t,0},\tbt_{\v t})\in\cM_{\delta,\lambda}(\v t)$
(see \eqref{eqn:mle-solution}), we take
\[
  \v g_\delta(\v t)
  :=\nabla_{\v t} h_{\delta,\lambda}
  \bigl(\v t,\wt{\beta}_{\v t,0},\tbt_{\v t}\bigr)
  \in \partial^{\mathrm{env}} f_{\delta,\lambda}(\v t),
\]
which coincides with the gradient $\nabla f_{\delta,\lambda}(\v t)$
whenever the minimum is attained and the minimizer is unique (see
Section~\ref{sec:uniqueness-inner-sol}).
Each iteration requires only (i)~solving the inner penalized GLM once
to obtain $\v g(\v t)$, and (ii)~selecting the $k$ smallest
components of $\v g(\v t)$ to form a vertex $\v s\in\{0,1\}^{p-m}$.
No projections are needed and the update is an explicit convex
combination, making the method straightforward to implement.
Details on how to employ existing GLM solvers to compute
$\v g_\delta(\v t)$ at a given $\v t$ are provided in
Section~\ref{sec:inner-solvers}.

\begin{algorithm}[h]
\caption{COMBSS-GLM}
\label{alg:fw-homotopy-joint}
\KwInput{Uniform penalty $\lambda\ge 0$,
  sparsity level $k$,
  initial curvature parameter $\delta_{\min}>0$, max curvature parameter $\delta_{\max} > \delta_{\min}$,
  grid size $N$, and learning rate $\alpha$}
\KwInit{$i\leftarrow 0$,
  $\v t\leftarrow \lt(\frac{k}{p-m}\rt)\v 1$, Set $r=(\delta_{\max}/\delta_{\min})^{1/N}$}
\For{$i = \{1, 2, \dots,  2N\}$}{
    $\delta \leftarrow \min\{\delta_{\min}r^i,\delta_{\max}\}$\;
  Compute $\v g_\delta(\v t)\in
    \partial^{\mathrm{env}} f_{\delta,\lambda}(\v t)$\;
  \label{step:grad}
  Let $\v s\in\{0,1\}^{p-m}$ select the $k$ smallest components of
    $\v g_\delta(\v t)$\;
  $\v t \leftarrow (1-\alpha)\v t + \alpha\,\v s$\;
}
Fit $(\beta_0,\v\beta)$ on the last vertex $\v s$\;
\KwOutput{$\v s$ and $\v\beta$.}
\end{algorithm}

The rationale for increasing $\delta$ gradually, rather than using a
fixed value, rests on three properties established in this paper.
First, by Proposition~\ref{prop:monotone-continuous-delta},
$f_{\delta,\lambda}(\v t)$ is continuous and non-decreasing in $\delta$
for each fixed $\v t$; in particular, the relaxed objective and its
minimisers change \emph{continuously} as $\delta$ varies.
Second, at $\delta=0$ with $\lambda=0$, the value function is constant
in $\v t$ (since the inner minimisation over $\bbt$ is equivalent, via
the substitution $\bbt_u\mapsto\v t\odot\bbt_u$, to an unconstrained
minimisation that is independent of $\v t$), hence trivially convex,
and the landscape admits effective gradient-based exploration.
Third, by Theorem~\ref{prop:concavity-threshold-unified}, once $\delta$
reaches $\delta_{\mathrm{conc}}$, $f_{\delta,\lambda}$ becomes concave
over $\cT_k$, guaranteeing that every minimiser is a binary corner
$\v s\in\cS_k$.
The geometric homotopy schedule therefore continuously deforms the
relaxed objective from this tractable small-$\delta$ regime to a
concave regime whose global minimisers are exactly the discrete subsets
sought (at $\delta_{\max}=\delta_{\mathrm{conc}}$), steering the
Frank--Wolfe iterates along this continuous path toward an optimal or
near-optimal binary corner.

Algorithm~\ref{alg:fw-homotopy-joint} initializes at the interior
point $\v t_0 := \frac{k}{p-m}\v 1$, the centroid of $\cT_k$, and
from a small curvature level $\delta_{\min}>0$, which is then increased
geometrically as $\delta_i=\delta_{\min}r^i$ with $r>1$.
At each iteration, we compute an envelope (sub)gradient
$\v g_{\delta}$ of the value function $f_{\delta,\lambda}$ at the
current $\v t$ and $\delta$.
The Frank--Wolfe update requires solving the \emph{linear minimization
oracle} (LMO) over the feasible set $\cT_k$.
Given the current envelope (sub)gradient $\v g_{\delta_i}(\v t)$, the
LMO is the linear program
\begin{equation}
  \label{eqn:lmo}
  \v s \in \argmin_{\v u\in \cT_k}\ \langle \v g_{\delta_i}(\v t),\,\v u\rangle.
\end{equation}
Because the objective in \eqref{eqn:lmo} is linear and $\cT_k$ is a
polytope, an optimal solution is attained at an extreme point of
$\cT_k$.
As noted earlier, these extreme points are precisely the binary
vertices in $\cS_k$; hence the LMO admits a binary solution
$\v s\in\cS_k$.
Moreover, for any $\v s\in\cS_k$ we have
$\langle \v g, \v s\rangle=\sum_{j:s_j=1} g_j$, so minimizing
$\langle \v g,\v s\rangle$ reduces to selecting the $k$ smallest
entries of $\v g_{\delta_i,\lambda}(\v t)$ (breaking ties
arbitrarily) and setting the corresponding components of $\v s$ to
$1$, with the rest set to $0$.
We then update $\v t$ by the explicit convex combination
$\v t\leftarrow (1-\alpha)\v t+\alpha\v s$, which moves $\v t$
toward this sparse vertex without requiring any projection.
The procedure terminates either after $N$ iterations or once $\v t$
is within $\varepsilon$ (in max-norm) of the current vertex,
indicating that $\v t$ is close to a binary corner.
Finally, we refit the GLM on the selected support $\v s$ using 
$X_{[\v s]}$ to obtain the
reported coefficients.
\subsection{Column Normalization and Invariance of the Selection Problem}
\label{sec:col-normalization}

For numerical stability, we apply a column-wise normalization to the
non-intercept columns of the design matrix.
This preprocessing does not change the subset-selection problem, since
it corresponds to a simple reparameterization.

Let $X_{:j}$ denote the $j$th column of $X$ for $j=0,1,\dots,p$,
where $X_{:0}=\v 1$ is the intercept column.
For $j=1,\dots,p$, define the column lengths $v_j := \|X_{:j}\|_2$,
and assume $v_j>0$ (columns with $v_j=0$ can be removed).
Set $v_0:=1$ and let $\v v:=(v_0,\dots,v_p)^\top$.
Define the normalized design
\begin{equation}
  \label{eqn:scaling}
  \wh X := X\,\diag(1/\v v),
\end{equation}
so that $\wh X_{:0}=X_{:0}$ and $\wh X_{:j}=X_{:j}/v_j$ for $j\ge 1$.

To see the invariance, recall that the log-likelihood depends on
predictors only through the linear predictor
$\beta_0+\sum_{j=1}^p x_j\beta_j$.
Since $v_j>0$, we have
\[
  \sum_{j=1}^p x_j\beta_j
  =
  \sum_{j=1}^p \Bigl(\frac{x_j}{v_j}\Bigr)\,(v_j\beta_j).
\]
Thus, defining $\widehat\beta_j:=v_j\beta_j$ (for $j\ge 1$) and
$\widehat x_j:=x_j/v_j$ leaves the linear predictor unchanged:
$\beta_0+\v x^\top\bbt=\beta_0+\widehat{\v x}^{\,\top}\widehat{\bbt}$.
Consequently, any likelihood-based subset-selection formulation
(including \eqref{eqn:dbss} and its relaxation) can be equivalently
posed using $\wh X$ in place of $X$, with a one-to-one correspondence
between coefficient vectors.
After the algorithm outputs a support $\v s$, we refit the
corresponding low-dimensional GLM on the \emph{original} design
matrix $X$ to report coefficients on the original scale.

\subsection{Fast Inner Fits and Gradients via Ridge Solvers}
\label{sec:inner-solvers}

Due to the Frank--Wolfe step updates, we compute $\v g(\v t)$ only at
interior points $\v t\in(0,1)^{p-m}$.
At such points, it is convenient to rewrite the inner objective so
that $\v t$ appears only through an $\cL_2$-penalty, enabling the use
of standard ridge-penalized GLM solvers.

Recall that the inner objective is
\[
  h_{\delta,\lambda}(\v t,\beta_0,\v\beta)
  =
  -\frac{1}{n}\ell(\beta_0,T_{\v t}\v\beta;\cD)
  +\lambda\|\v\beta\|_2^2
  +\delta\|\Gamma_{\v t}\v\beta\|_2^2.
\]
For an interior point $\v t\in(0,1)^{p-m}$, the change of variables
$\xi_0:=\beta_0$ and $\v\xi := T_{\v t}\v\beta$ gives
$\v\beta=T_{\v t}^{-1}\v\xi$ and
\begin{align}
  \phi_{\delta,\lambda}(\v t,\xi_0,\v\xi)
  &:=
  -\frac{1}{n}\ell(\xi_0,\v\xi;\cD)
  +\lambda\|T_{\v t}^{-1}\v\xi\|_2^2
  +\delta\|\Gamma_{\v t}T_{\v t}^{-1}\v\xi\|_2^2 \nonumber\\
  &=
  -\frac{1}{n}\ell(\xi_0,\v\xi;\cD)
  +\Bigl\|\bigl(\lambda I+\delta\,\Gamma_{\v t}^2\bigr)^{1/2}
  T_{\v t}^{-1}\v\xi\Bigr\|_2^2,
  \label{eqn:phi-def}
\end{align}
so that $f_{\delta,\lambda}(\v t)
=\min_{\xi_0\in\reals,\;\v\xi\in\reals^p}
\phi_{\delta,\lambda}(\v t,\xi_0,\v\xi)$.
Using the envelope theorem from \citet{ruszczynski2011nonlinear},
\[
  \nabla f_{\delta,\lambda}(\v t)
  =\nabla_{\v t}\phi_{\delta,\lambda}(\v t,\wt\xi_0,\wt{\v\xi}),
\]
where $(\wt\xi_0,\wt{\v\xi})$ is any minimizer.
Since the negative log-likelihood term does not depend on $\v t$, and
the penalty is coordinate-separable,
\begin{align}
  \Bigl\|\bigl(\lambda I+\delta\,\Gamma_{\v t}^2\bigr)^{1/2}
  T_{\v t}^{-1}\v\xi\Bigr\|_2^2
  &=
  \lambda\sum_{j=1}^{m}\xi_j^2
  +\sum_{j=1}^{p-m}\left(\frac{\lambda+\delta}{t_j^{2}}-\delta\right)
  \xi_{m+j}^2,
  \label{eqn:penalty-expansion-vec}
\end{align}
it follows that for each $j=1,\dots,p-m$,
\[
  \frac{\partial}{\partial t_j}
  \phi_{\delta,\lambda}(\v t,\xi_0,\v\xi)
  =
  -2(\lambda+\delta)\,\frac{\xi_{m+j}^2}{t_j^{3}},
\]
and hence
\[
  \nabla_{\v t} \phi_{\delta,\lambda}(\v t,\xi_0,\v\xi)
  =
  -2(\lambda+\delta)\,
  \begin{pmatrix}
    \xi_{m+1}^2/t_1^3\\ \vdots\\ \xi_{p}^2/t_{p-m}^3
  \end{pmatrix}.
\]

\noindent{\bf Using solvers without coefficient-wise penalties.}
While packages such as \texttt{glmnet} in {\sf R} allow
coefficient-wise ridge weights, many solvers expose only a uniform
$\cL_2$-penalty.
To take advantage of such solvers, fix an interior $\v t$ and define
\[
  \omega_j(\v t):=
  \begin{cases}
    \lambda, & j=1,\dots,m,\\[2pt]
    \displaystyle \frac{\lambda+\delta}{t_{j-m}^{2}}-\delta,
    & j=m+1,\dots,p.
  \end{cases}
\]
Introduce the reparameterization
$\v\theta=\sqrt{\v\omega(\v t)}\odot\v\xi$,
i.e., $\theta_j=\sqrt{\omega_j(\v t)}\,\xi_j$, and scale the design
columns accordingly using
$\widetilde X_{\v t}:=X\,\diag\!\bigl(\v\omega(\v t)^{-1/2}\bigr)$.
Then $X\v\xi=\widetilde X_{\v t}\v\theta$ and
$\sum_{j=1}^p \omega_j(\v t)\,\xi_j^2=\|\v\theta\|_2^2$, so the
inner problem can be solved with any ridge GLM solver on
$(\widetilde X_{\v t},\v y)$ using a uniform $\cL_2$-penalty.
After solving for $\wt{\v\theta}$, recover
$\wt{\v\xi}=\v\omega(\v t)^{-1/2}\odot \wt{\v\theta}$.

\subsection{Uniqueness of the Inner Minimizer on the Interior}
\label{sec:uniqueness-inner-sol}

Fix an interior point $\v t\in(0,1)^{p-m}$.

\paragraph{No mandatory (unpenalized) variables}
If $m=0$, then for any interior point $\v t\in(0,1)^p$ we have
$\Gamma_{\v t}^2=I-T_{\v t}^2$ with strictly positive diagonal, and
\[
  \lambda\|\bbt\|_2^2+\delta\|\Gamma_{\v t}\bbt\|_2^2
  =
  \sum_{j=1}^{p}\bigl(\lambda+\delta(1-t_j^2)\bigr)\beta_j^2,
\]
so $h_{\delta,\lambda}(\v t,\beta_0,\bbt)$ is \emph{strongly convex}
in $\bbt$ (even when $\lambda=0$, since $\delta>0$ and
$t_j\in(0,1)$).
Consequently, whenever $-\ell(\beta_0,\bbt;\cD)$ is convex in
$(\beta_0,\bbt)$ and the minimum is attained, the coefficient
minimizer is unique.
Under standard additional conditions ensuring strict convexity in the
intercept direction, the pair $(\wt\beta_{\v t,0},\tbt_{\v t})$ is
unique.
The same conclusions hold for the reparameterized inner objective
$\phi_{\delta,\lambda}(\v t,\xi_0,\v\xi)$, since for interior $\v t$
the change of variables $(\xi_0,\v\xi)=(\beta_0,T_{\v t}\bbt)$ is
one-to-one.

\paragraph{Intercept and $m$ mandatory variables}
When an intercept and $m$ mandatory coefficients are left unscaled by
$\v t$, the quadratic term splits as
\[
  \lambda\|\bbt\|_2^2+\delta\|\Gamma_{\v t}\bbt\|_2^2
  =
  \lambda\sum_{j=1}^{m}\beta_j^2
  +\sum_{j=1}^{p-m}\bigl(\lambda+\delta(1-t_j^2)\bigr)\beta_{m+j}^2,
\]
so the non-mandatory coefficients always carry a strictly positive
quadratic weight on the interior.
Existence of a minimizer follows, for example, if
$-\ell(\beta_0,\bbt;\cD)$ is convex and coercive in the unpenalized
directions (intercept and any mandatory coefficients); see, e.g.,
\citep[Chapters~2--3]{mccullagh2019generalized}.
Uniqueness holds whenever the overall inner objective is strictly
convex in $(\beta_0,\bbt)$; a sufficient condition is strict
convexity of $-\ell(\beta_0,\bbt;\cD)$, while taking $\lambda>0$
ensures strong convexity in all mandatory coefficients (and in all
coefficients if $m=0$).

\paragraph{Non-unique inner minimizer}
Even if the inner minimizer is not unique, the algorithm remains
well-defined.
Since for every fixed $(\xi_0,\v\xi)$ the map
$\v t\mapsto \phi_{\delta,\lambda}(\v t,\xi_0,\v\xi)$ is continuously
differentiable on $(0,1)^{p-m}$, Danskin--Clarke envelope results
apply to the value function
$f_{\delta,\lambda}(\v t)
=\min_{\xi_0,\v\xi}\phi_{\delta,\lambda}(\v t,\xi_0,\v\xi)$,
and the subdifferential satisfies
\[
  \partial f_{\delta,\lambda}(\v t)
  =
  \operatorname{conv}\Bigl\{
  \nabla_{\v t}\phi_{\delta,\lambda}(\v t,\xi_0,\v\xi)
  :(\xi_0,\v\xi)\in\cM(\v t)\Bigr\}.
\]
Hence, for any $(\wt\xi_0,\wt{\v\xi})\in\cM(\v t)$, the vector
$\v g(\v t) =\nabla_{\v t}\phi_{\delta,\lambda}(\v t,\wt\xi_0,\wt{\v\xi})$
is a valid generalized subgradient for use in Step~\ref{step:grad}.
When $\cM(\v t)$ is a singleton, this reduces to the ordinary gradient
$\v g(\v t)=\nabla f_{\delta,\lambda}(\v t)$.

\subsection{Properties of the Objective Function}
\label{sec:fw-convergence}
We now establish continuity and large-$\delta$ concavity properties of the relaxed objective $f_{\delta,\lambda}(\v t)$. Proofs of all the results reported here are in the supplemetary material.

\begin{proposition}[Monotonicity and continuity in $\delta$]
\label{prop:monotone-continuous-delta}
For every fixed $\v t\in(0,1)^{p-m}$, the map
$\delta\mapsto f_{\delta,\lambda}(\v t)$ is 
monotone non-decreasing 
and continuous in $\delta>0$.
\end{proposition}

Let $X_u\in\reals^{n\times (p-m)}$ denote the submatrix of
non-mandatory columns.
For fixed $(\beta_0,\bbt)$ and interior $\v t$, write
$\bbt_u\in\reals^{p-m}$ for the non-mandatory coefficients and
$D_{\bbt_u}:=\diag(\bbt_u)$.
Define $H_\eta(\eta):=\nabla_{\eta\eta}^2\!\left(-\frac{1}{n}\ell\right)(\eta)$,
the Hessian of $-\frac{1}{n}\ell$ with respect to the linear
predictor $\eta$.
Let
$\eta(\v t,\beta_0,\bbt):=\eta_0(\beta_0,\bbt)+X_u(\v t\odot \bbt_u)$,
where $\eta_0(\beta_0,\bbt)$ collects the intercept and
mandatory-feature contribution.
Then
\[
  \nabla_{\v t\v t}^2 h_{\delta,\lambda}(\v t,\beta_0,\bbt)
  =
  D_{\bbt_u}\,X_u^\top H_\eta\!\big(\eta(\v t,\beta_0,\bbt)\big)
  X_u\,D_{\bbt_u}
  -2\delta D_{\bbt_u}^2,
\]
where the ridge term $\lambda\|\bbt\|_2^2$ does not contribute as it
is constant in $\v t$.

\begin{theorem}[Concavity threshold in $\v t$]
\label{prop:concavity-threshold-unified}
Assume $X_u^\top H_\eta(\eta) X_u \preceq 2\delta I$ for all $\eta$.
Then, for every fixed $(\beta_0,\bbt)$, the map
$\v t\mapsto h_{\delta,\lambda}(\v t,\beta_0,\bbt)$ is concave on
$(0,1)^{p-m}$.
Consequently, $f_{\delta,\lambda}(\v t)
=\inf_{\beta_0,\bbt} h_{\delta,\lambda}(\v t,\beta_0,\bbt)$
is concave on $(0,1)^{p-m}$ (and extends concavely to $\cT_k$), so
every minimizer of $\min_{\v t\in\cT_k} f_{\delta,\lambda}(\v t)$
can be chosen at a corner $\v s\in\cS_k$.
\end{theorem}

For logistic and baseline multinomial regression, the standard
curvature bounds
\[
  H_\eta(\eta)\preceq
  \begin{cases}
    \tfrac{1}{4n}I, & \text{logistic},\\[3pt]
    \tfrac{1}{2n}I, & \text{multinomial},
  \end{cases}
\]
yield the explicit sufficient thresholds.  Writing
$\nu_{\max}$ for the largest eigenvalue of $X_u^\top X_u$,
\begin{align}
  \delta_{\mathrm{conc}}=
  \begin{cases}
    \nu_{\max}/(8n), & \text{logistic},\\
    \nu_{\max}/(4n), & \text{multinomial}.
  \end{cases}
  \label{eqn:delta_conc}
\end{align}

\begin{remark}
The Hessian formula above is stated for the binary case in which
$\bbt_u \in \reals^{p-m}$ is a vector.
In the multinomial setting, $B_u \in \reals^{(p-m)\times(C-1)}$ is a
matrix and the Hessian of $h$ with respect to $\v t$ sums contributions
from each of the $C-1$ class-specific coefficient vectors.
The bound $H_\eta(\eta) \preceq \frac{1}{2n}I$ follows from the fact that
the per-observation multinomial curvature matrix
$\Sigma_i = \diag(\pi_i) - \pi_i\pi_i^\top \in \reals^{(C-1)\times(C-1)}$
has largest eigenvalue at most $\tfrac{1}{2}$, as can be shown
via the Gershgorin circle theorem (replacing the logistic bound
$\pi_i(1-\pi_i) \le \tfrac{1}{4}$).
Theorem~\ref{prop:concavity-threshold-unified} and the threshold
$\nu_{\max}/(4n)$ hold for the multinomial case by the same proof.
\end{remark}

\begin{remark}[Computational complexity]
\label{rem:complexity}
Algorithm~\ref{alg:fw-homotopy-joint} executes $2N$ iterations, each
requiring one call to a penalised GLM solver plus $O(p)$ overhead
(gradient extraction via Danskin's theorem and a partial sort to identify
the $k$ smallest components).  The total cost is thus
$2N \times T_{\mathrm{solver}}(n, p, C)$, where
$T_{\mathrm{solver}}(n,p,C)$ is the cost of a single penalised GLM fit.
With a modest budget such as $N = 25$, this is comparable to fitting a
regularisation path with $50$ penalty values.
\end{remark}

\section{Simulation Studies}
\label{s:simulations}

We evaluate the performance of the proposed best subset selection
method for logistic regression through a series of simulation
experiments in both low-dimensional and high-dimensional settings.
The simulation design closely follows the linear model framework
established in \citet{hastie2020best} and \citet{moka2024combss}, adapted here to the
binary response setting.

\subsection{Simulation Design}
\label{ss:sim-design}

Data are generated from a logistic regression model.
Specifically, the binary response $y_i \in \{0,1\}$,
$i = 1,\ldots,n$, is drawn from a Bernoulli distribution with success
probability
\begin{equation}
  \label{eq:logistic_model}
  P(y_i = 1 \mid \v x_i) =
  \frac{\exp(\beta_0+\v x_i^\top \boldsymbol{\beta}^*)}%
       {1 + \exp(\beta_0+\v x_i^\top \boldsymbol{\beta}^*)},
\end{equation}
where $\boldsymbol{\beta}^* \in \mathbb{R}^p$ is the true coefficient
vector.
The intercept is fixed at $\beta_0 = 0.2$ to ensure an approximately
balanced proportion of zeros and ones in the simulated responses.
Each row $\v x_i$ of the predictor matrix $X$ is generated
independently from a multivariate normal distribution with zero mean
and covariance matrix $\boldsymbol{\Sigma}$ with diagonal elements
$\Sigma_{j,j} = 1$ and off-diagonal elements
$\Sigma_{i,j} = \rho^{|i-j|}$, $i \neq j$.
We vary the correlation parameter over
$\rho \in \{0,\, 0.2,\, 0.4,\, 0.6\}$ to examine the effect of
predictor correlation on the performance of the methods. In these simulations, no mandatory variables are kept $(m = 0)$.

We consider two configurations for the true coefficient vector
$\boldsymbol{\beta}^*$:
\begin{itemize}
  \item Case~1: The first $k_0 = 10$ components of
    $\boldsymbol{\beta}^*$ are equal to $1$ and all other components
    are equal to $0$.
  \item Case~2: The first $k_0 = 10$ components of
    $\boldsymbol{\beta}^*$ decay as $\beta_i^* = 0.5^{\,i-1}$ for
    $i = 1,\ldots,k_0$, while all other components are equal to $0$.
\end{itemize}
For the low-dimensional setting we take $n = 200$ and $p = 30$, and
for the high-dimensional setting $n = 200$ and $p = 1000$.
For each combination of case, dimension, and correlation level, we
perform 50 independent replications.

\subsection{Competing Methods and Tuning}
\label{ss:sim-methods}

We compare the proposed COMBSS-GLM method
(Algorithm~\ref{alg:fw-homotopy-joint}) against three penalized
likelihood methods: the Lasso \citep{Tibshirani96}, the smoothly
clipped absolute deviation penalty (SCAD) \citep{fan2001scad}, and
the minimax concave penalty (MCP) \citep{zhang2010mcp}, all available
through the \texttt{ncvreg} package in {\sf R}
\citep{breheny2011ncvreg}.

Algorithm~\ref{alg:fw-homotopy-joint} is executed for each candidate
model size $k = 1, \ldots, 20$ using $N=25$ iterations and a learning rate $\alpha=0.01$.
The optimal model size $k_{\mathrm{opt}}$ is then selected using an
independent test set of $10{,}000$ observations generated from the
same data-generation mechanism, by choosing the model size that yields
the lowest misclassification error, where predicted class labels are
obtained by thresholding the estimated probabilities at $0.5$.

For the Lasso, SCAD, and MCP, we fit each method over a grid of 50
regularization parameters $\lambda$, log-spaced from $\lambda_{\max}$
down to a small fraction of $\lambda_{\max}$
as implemented in \texttt{glmnet} package.
For each method, we choose $\lambda$ to maximize classification
accuracy on the same independent test set of $10{,}000$ observations,
and report the corresponding selected variable subset as the final
model.

\subsection{Performance Metrics}
\label{ss:sim-metrics}

We assess the methods using several performance metrics.
To evaluate \emph{variable selection quality}, we compare the
selected set $\hat{S}$ to the true active set
$S_0 = \{j : \beta_j^* \neq 0\}$ and report:
\begin{itemize}
  \item Sensitivity (true positive rate):
    $|\hat{S} \cap S_0| / |S_0|$.
  \item Specificity (true negative rate):
    $|\hat{S}^c \cap S_0^c| / |S_0^c|$, where
    $S^c = \{1,\ldots,p\}\setminus S$.
  \item Selection accuracy: proportion of predictors whose
    inclusion status matches the truth,
    $(|\hat{S} \cap S_0| + |\hat{S}^c \cap S_0^c|) / p$.
  \item F1 score: harmonic mean of precision and sensitivity.
  \item Matthews correlation coefficient (MCC): a balanced
    measure of binary classification quality for the selected variable
    set, which accounts for all four entries of the confusion matrix
    and remains informative even when class sizes are unequal
    \citep{chicco2020mcc}.
\end{itemize}
In addition, we report \emph{prediction accuracy} on the test set,
defined as the proportion of test observations for which the predicted
class label $\hat{y}_i = \mathbb{I}(\hat{p}_i > 0.5)$ agrees with
the true response $y_i$.
All metrics are averaged over the 50 replications, and standard errors
are reported.

\subsection{Results}
\label{ss:sim-results}

\subsubsection{Low-Dimensional Setting ($n = 200$, $p = 30$)}

Figure~\ref{fig:sim_low} presents the results for Case~1 (top row)
and Case~2 (bottom row) in the low-dimensional setting across the four
levels of predictor correlation. In Case~1, where all active predictors share the same coefficient
magnitude, COMBSS-GLM achieves the highest MCC and F1 scores across all
correlation levels, demonstrating its ability to accurately recover
the true sparse model.
SCAD and MCP exhibit competitive performance at low correlation
($\rho = 0$ and $\rho = 0.2$), with MCC and F1 scores close to those
of COMBSS-GLM.
However, as the correlation increases to $\rho = 0.4$ and
$\rho = 0.6$, the performance of all penalized methods degrades more
sharply than that of COMBSS-GLM, which maintains relatively stable
variable selection accuracy.
The Lasso consistently yields the lowest MCC and F1 scores among all
methods, reflecting its well-known tendency to select overly dense
models in the presence of correlated predictors.
In terms of prediction accuracy, all four methods perform comparably,
though COMBSS-GLM achieves a slight edge at higher correlation levels.

In Case~2, where the true coefficients decay exponentially and thus
present a more challenging signal structure, the advantage of COMBSS-GLM
becomes more pronounced.
COMBSS-GLM consistently achieves the highest MCC and F1 scores across all
values of $\rho$, with the gap widening as the correlation increases.
Notably, the MCC of the Lasso, SCAD, and MCP drop substantially for
$\rho \geq 0.4$, whereas COMBSS-GLM deteriorates more gradually.
The prediction accuracy of all the methods in Case~2 is lower than in Case~1, as
expected given the weaker signal among the active predictors.
Nevertheless, COMBSS-GLM maintains a consistent advantage, particularly
at higher correlation levels.

Additional variable selection metrics are presented in
Figure~\ref{fig:supp_sim_low} in the supplementary material.
These results reveal that the Lasso achieves high sensitivity across
all settings but at the cost of substantially lower specificity,
confirming its tendency to include many false positives.
SCAD and MCP offer a better balance between sensitivity and specificity
than the Lasso, but COMBSS-GLM achieves the best overall trade-off, as
reflected in its superior selection accuracy across both cases.

\begin{figure}[t]
  \centering
  \begin{subfigure}[b]{0.32\textwidth}
    \centering
    \includegraphics[width=\textwidth]{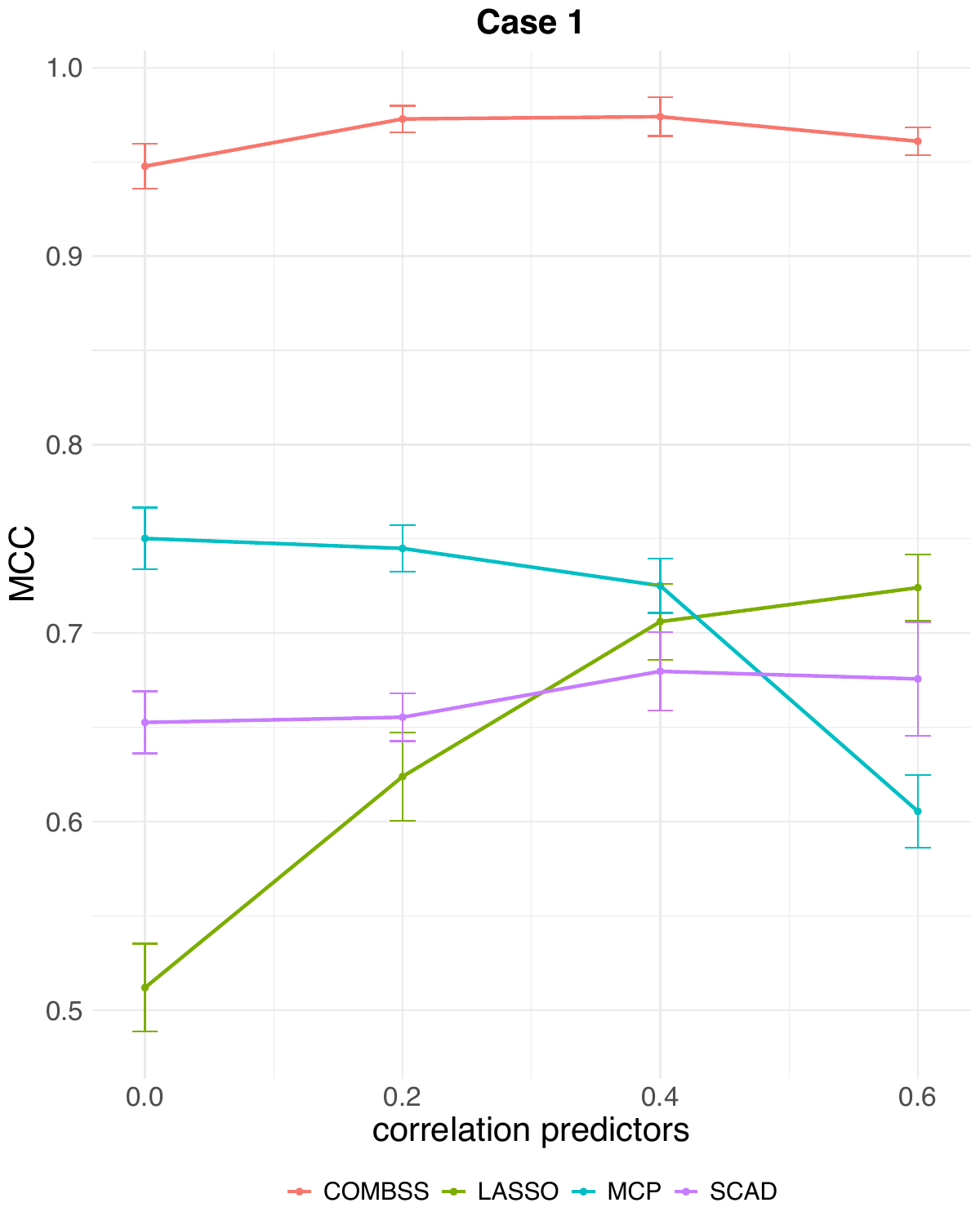}
    \caption{MCC -- Case 1}
    \label{fig:low_case1_mcc}
  \end{subfigure}
  \hfill
  \begin{subfigure}[b]{0.32\textwidth}
    \centering
    \includegraphics[width=\textwidth]{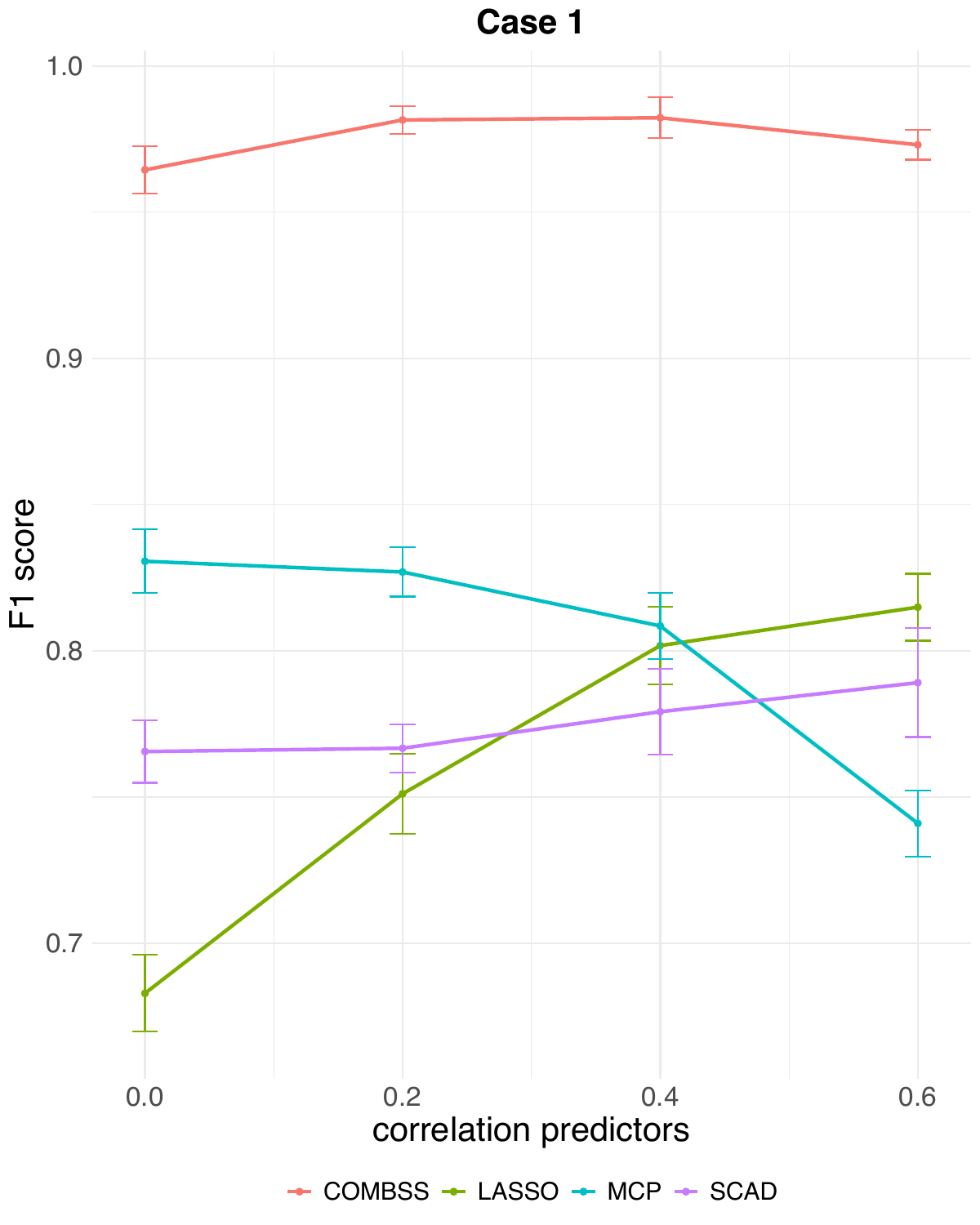}
    \caption{F1 score -- Case 1}
    \label{fig:low_case1_f1}
  \end{subfigure}
  \hfill
  \begin{subfigure}[b]{0.32\textwidth}
    \centering
    \includegraphics[width=\textwidth]{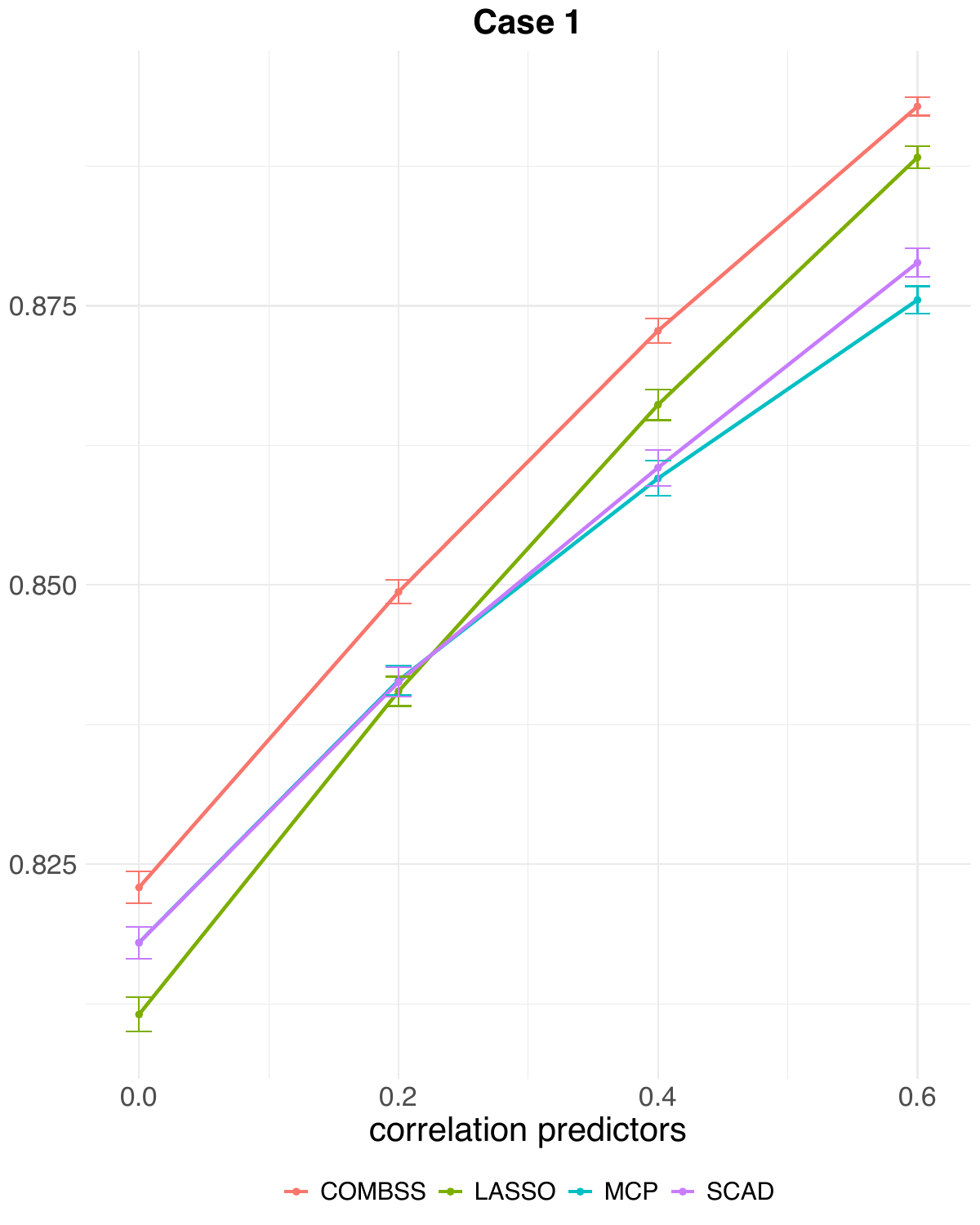}
    \caption{Prediction accuracy -- Case 1}
    \label{fig:low_case1_acc}
  \end{subfigure}

  \vspace{0.4cm}

  \begin{subfigure}[b]{0.32\textwidth}
    \centering
    \includegraphics[width=\textwidth]{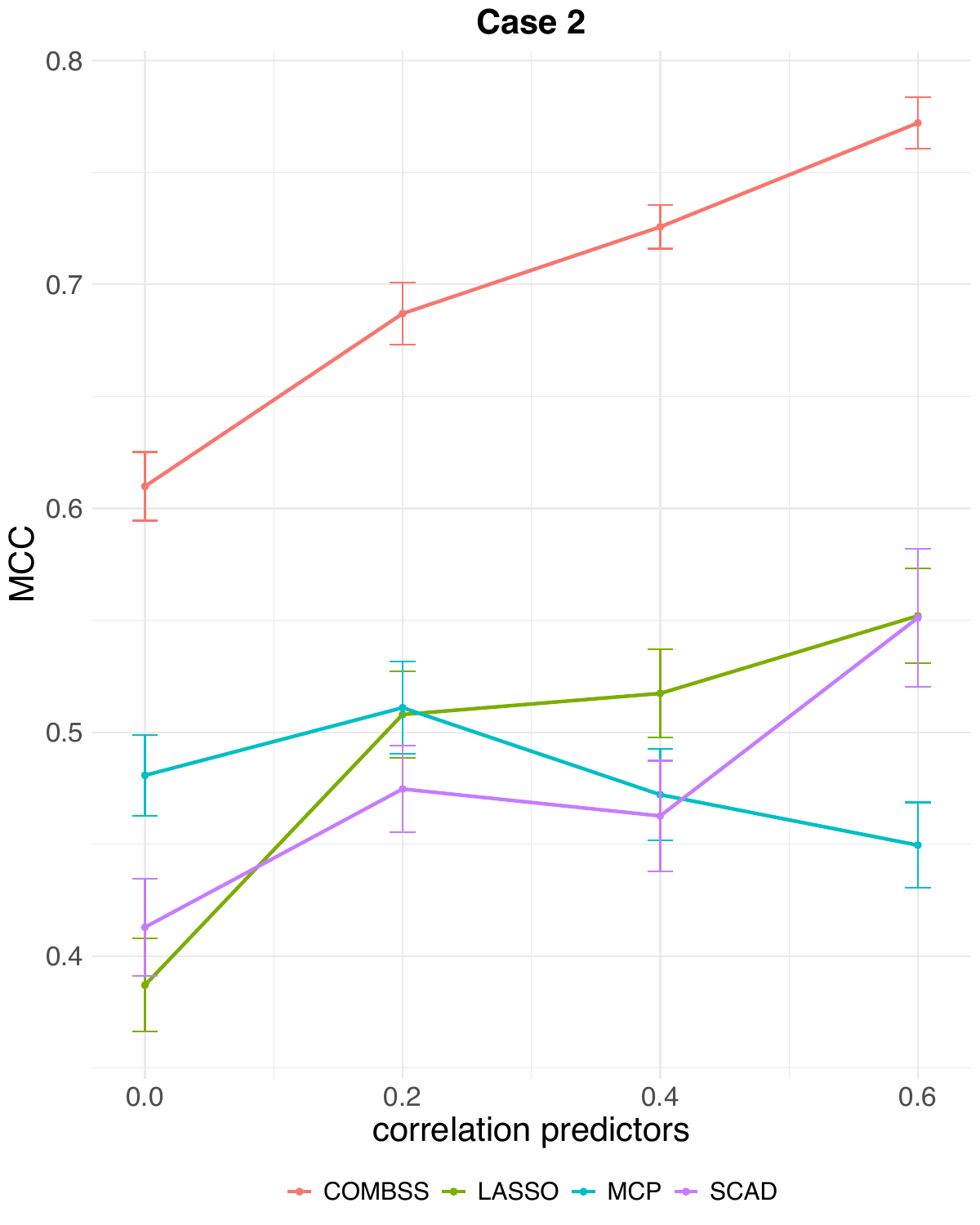}
    \caption{MCC -- Case 2}
    \label{fig:low_case2_mcc}
  \end{subfigure}
  \hfill
  \begin{subfigure}[b]{0.32\textwidth}
    \centering
    \includegraphics[width=\textwidth]{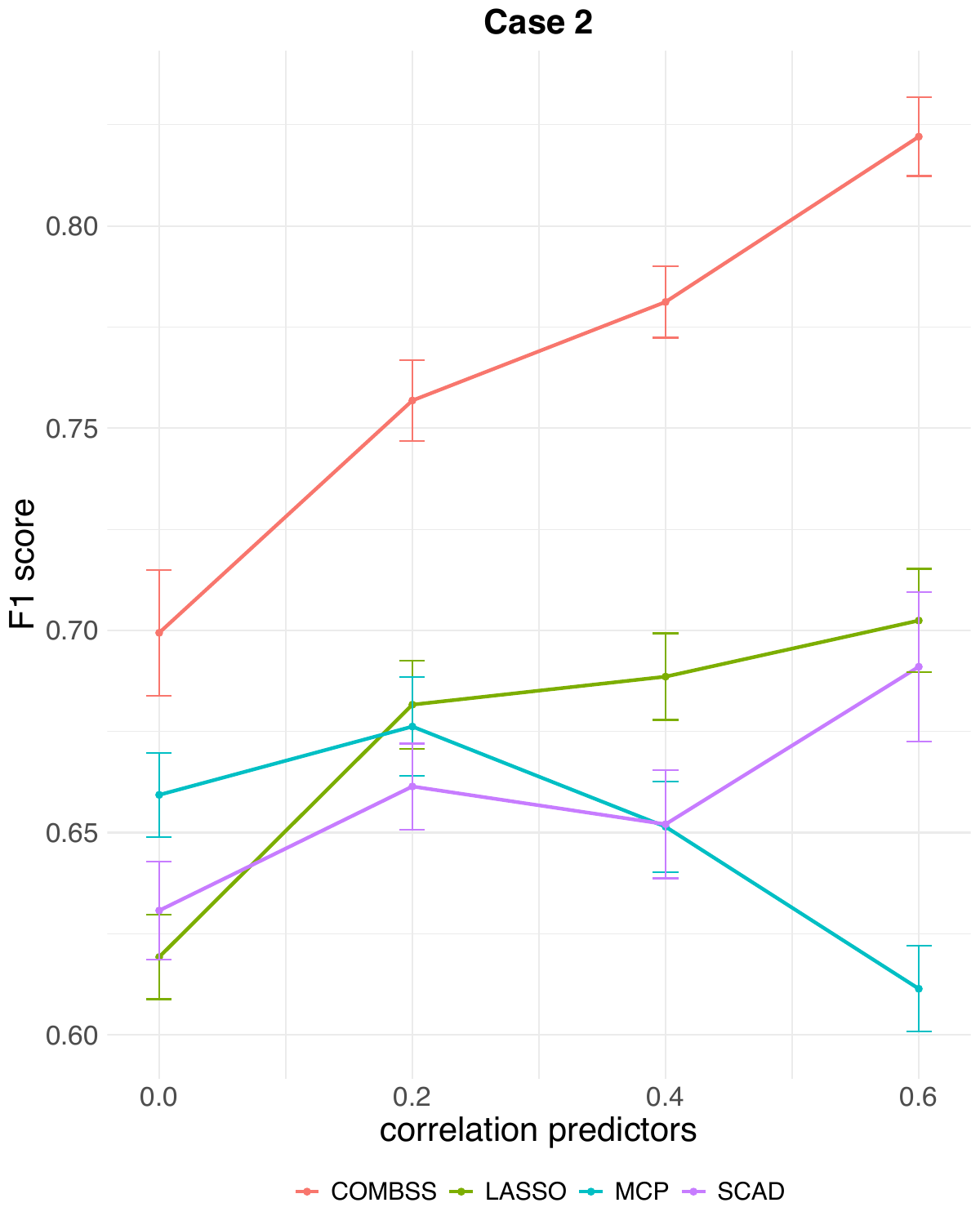}
    \caption{F1 score -- Case 2}
    \label{fig:low_case2_f1}
  \end{subfigure}
  \hfill
  \begin{subfigure}[b]{0.32\textwidth}
    \centering
    \includegraphics[width=\textwidth]{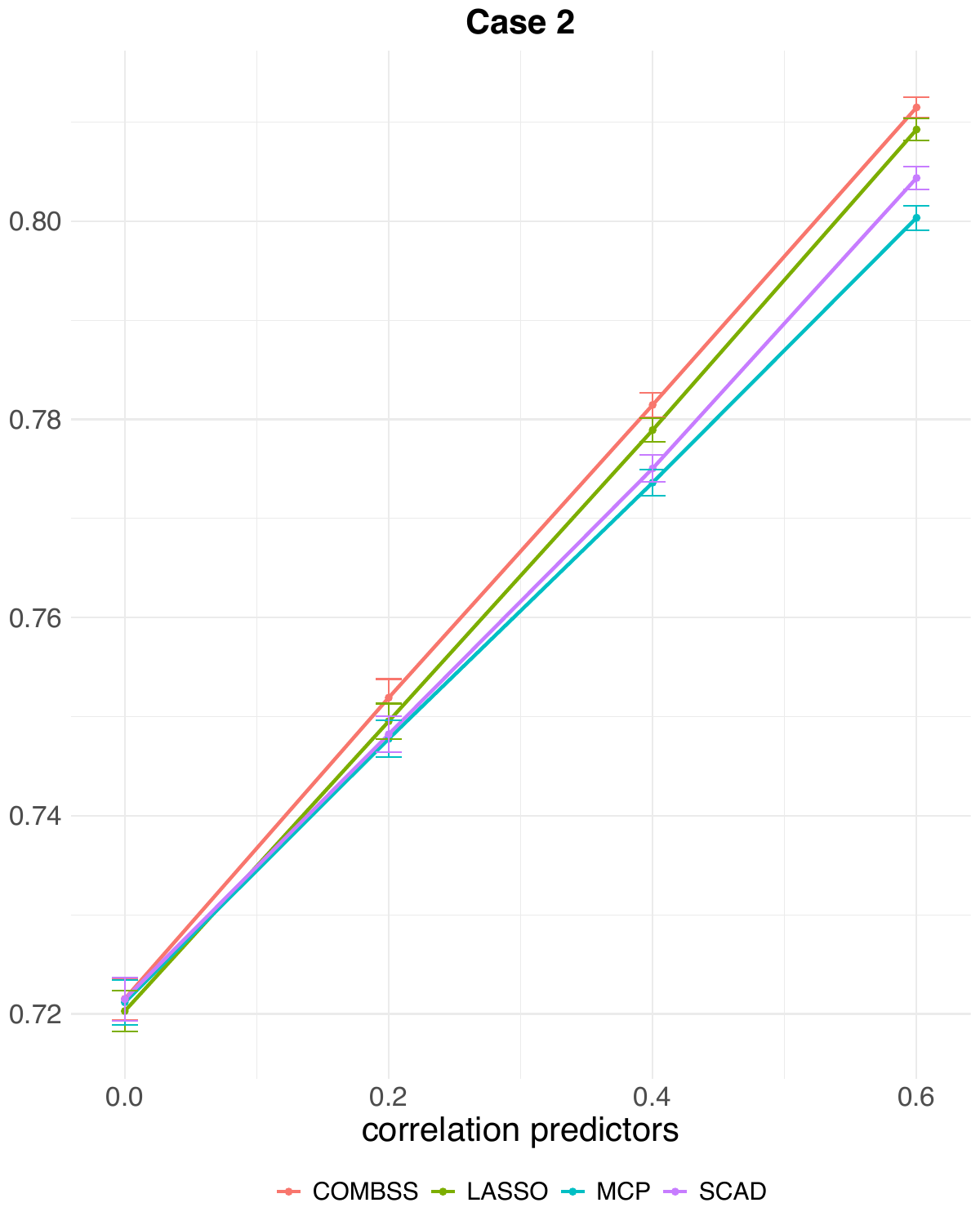}
    \caption{Prediction accuracy -- Case 2}
    \label{fig:low_case2_acc}
  \end{subfigure}
  \caption{Performance results in the low-dimensional setting
    ($n = 200$, $p = 30$) for Case~1 (top row) and Case~2 (bottom
    row). Each panel displays the average over 50 replications as a
    function of predictor correlation $\rho \in \{0, 0.2, 0.4, 0.6\}$,
    with vertical bars denoting one standard error.}
  \label{fig:sim_low}
\end{figure}

\subsubsection{High-Dimensional Setting ($n = 200$, $p = 1000$)}

Figure~\ref{fig:sim_high} presents the corresponding results in the
high-dimensional setting, where the number of predictors greatly
exceeds the sample size. In Case~1, COMBSS-GLM clearly outperforms all three penalized methods in
terms of MCC and F1 score across all values of $\rho$.
The relative advantage of COMBSS-GLM is substantially larger than in the
low-dimensional setting, highlighting the benefit of the best subset
selection approach when the model space is vast.
SCAD and MCP perform similarly to each other and offer an improvement
over the Lasso in terms of MCC and F1, but neither matches COMBSS-GLM.
Prediction accuracy remains comparable across all methods for low
correlation, but COMBSS-GLM demonstrates a more consistent advantage as
$\rho$ increases.

In Case~2, the performance advantage of COMBSS-GLM over its competitors is even more pronounced.
The MCC and F1 scores of the Lasso, SCAD, and MCP all decline
substantially with increasing correlation, while COMBSS-GLM maintains
notably higher values.
This is consistent with the finding from the linear model setting
\citep{moka2024combss}, where COMBSS-GLM was found to be particularly
effective at recovering models with heterogeneous signal strengths.
Prediction accuracy follows a similar pattern, with COMBSS-GLM achieving
the highest values across all correlation levels.

The supplementary metrics in Figure~\ref{fig:supp_sim_high} further
confirm these findings.
In the high-dimensional setting, the Lasso exhibits markedly lower
specificity than the other methods, as it tends to include a large
number of inactive predictors.
COMBSS-GLM achieves the highest sensitivity while maintaining high
specificity, resulting in the best overall selection accuracy.
SCAD and MCP demonstrate high specificity but at the cost of lower
sensitivity compared to COMBSS-GLM, particularly in Case~2 where the
weaker signals are more difficult to detect.

Overall, across both dimensional settings and both cases, COMBSS-GLM
consistently provides the best variable selection performance as
measured by MCC and F1 score, while achieving competitive or superior
prediction accuracy.
The advantage of COMBSS-GLM is most pronounced in the high-dimensional
setting and under stronger predictor correlation, precisely the
scenarios where variable selection is most challenging.
These results extend the findings of \citet{moka2024combss} for the
linear model to the logistic regression setting, and demonstrate the
practical value of COMBSS-GLM for
binary response data.

\begin{figure}[t]
  \centering
  \begin{subfigure}[b]{0.32\textwidth}
    \centering
    \includegraphics[width=\textwidth]{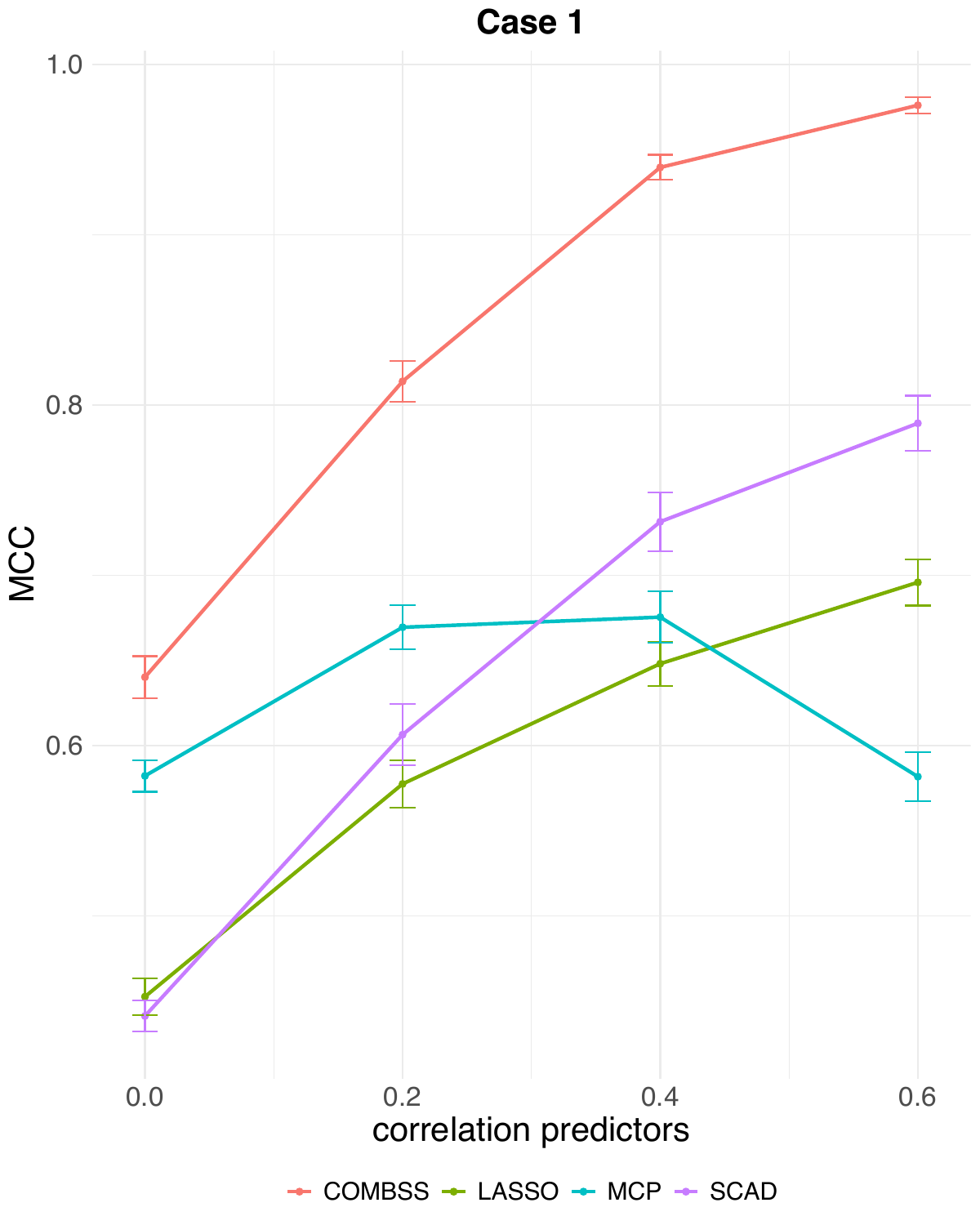}
    \caption{MCC -- Case 1}
    \label{fig:high_case1_mcc}
  \end{subfigure}
  \hfill
  \begin{subfigure}[b]{0.32\textwidth}
    \centering
    \includegraphics[width=\textwidth]{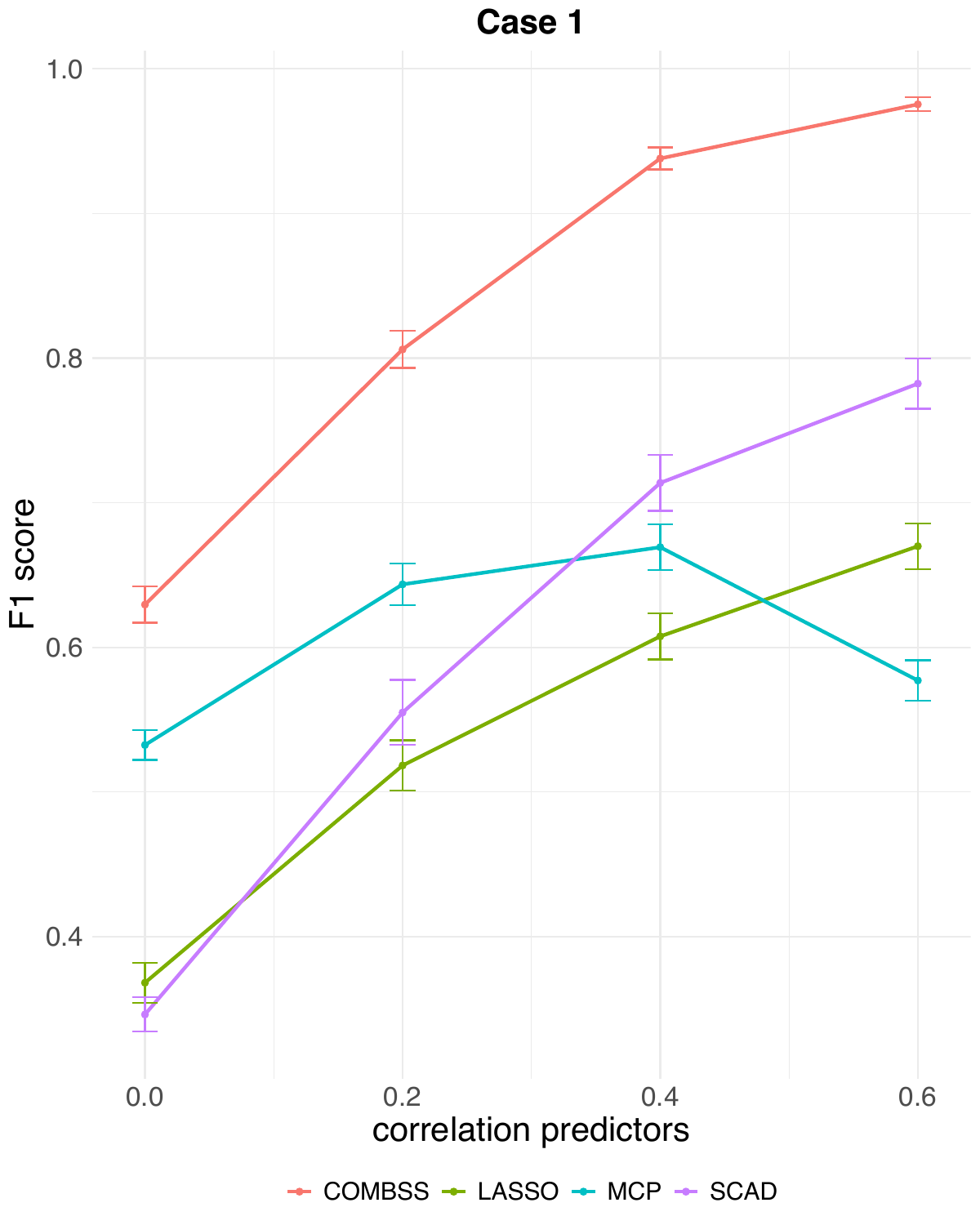}
    \caption{F1 score -- Case 1}
    \label{fig:high_case1_f1}
  \end{subfigure}
  \hfill
  \begin{subfigure}[b]{0.32\textwidth}
    \centering
    \includegraphics[width=\textwidth]{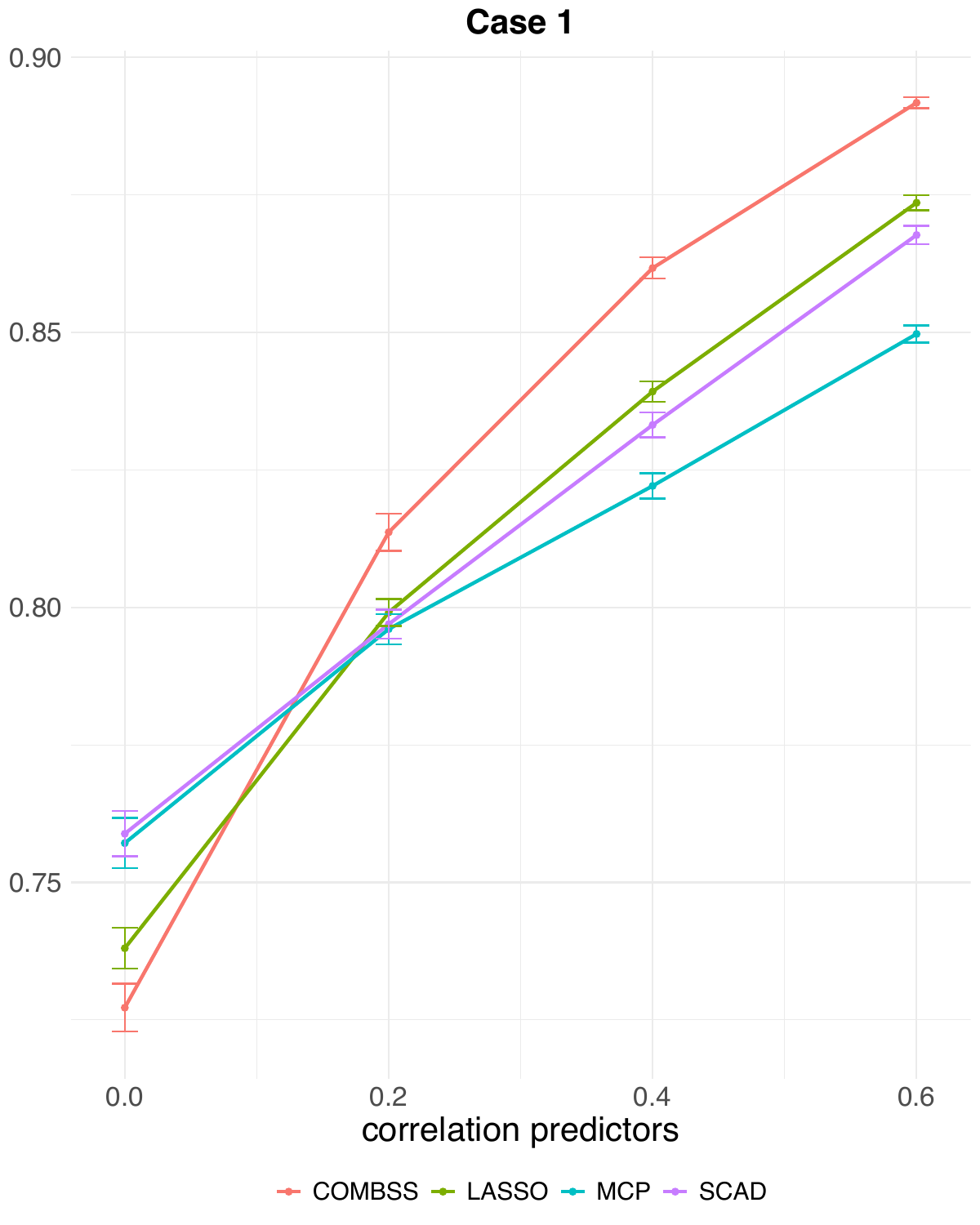}
    \caption{Prediction accuracy -- Case 1}
    \label{fig:high_case1_acc}
  \end{subfigure}

  \vspace{0.4cm}

  \begin{subfigure}[b]{0.32\textwidth}
    \centering
    \includegraphics[width=\textwidth]{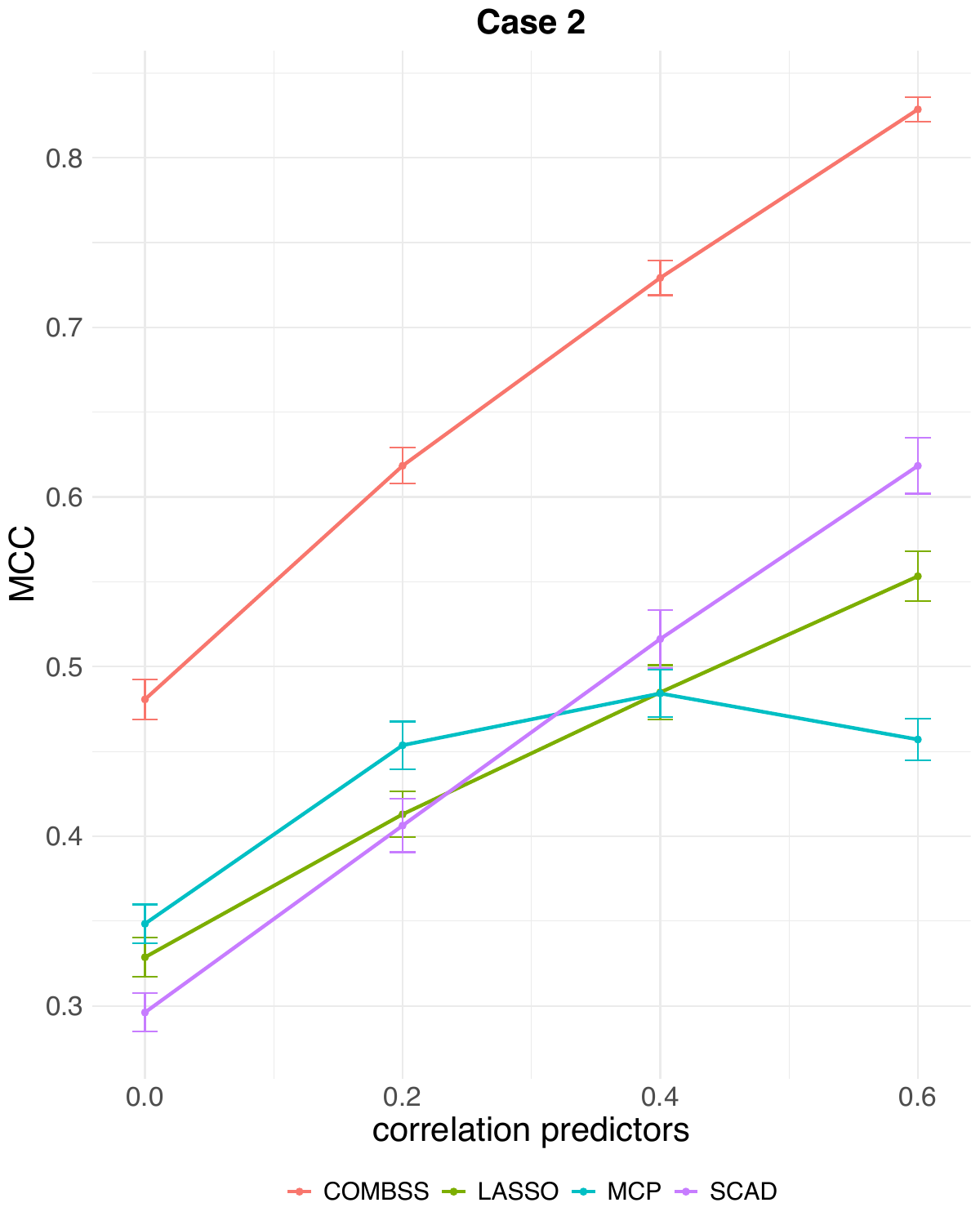}
    \caption{MCC -- Case 2}
    \label{fig:high_case2_mcc}
  \end{subfigure}
  \hfill
  \begin{subfigure}[b]{0.32\textwidth}
    \centering
    \includegraphics[width=\textwidth]{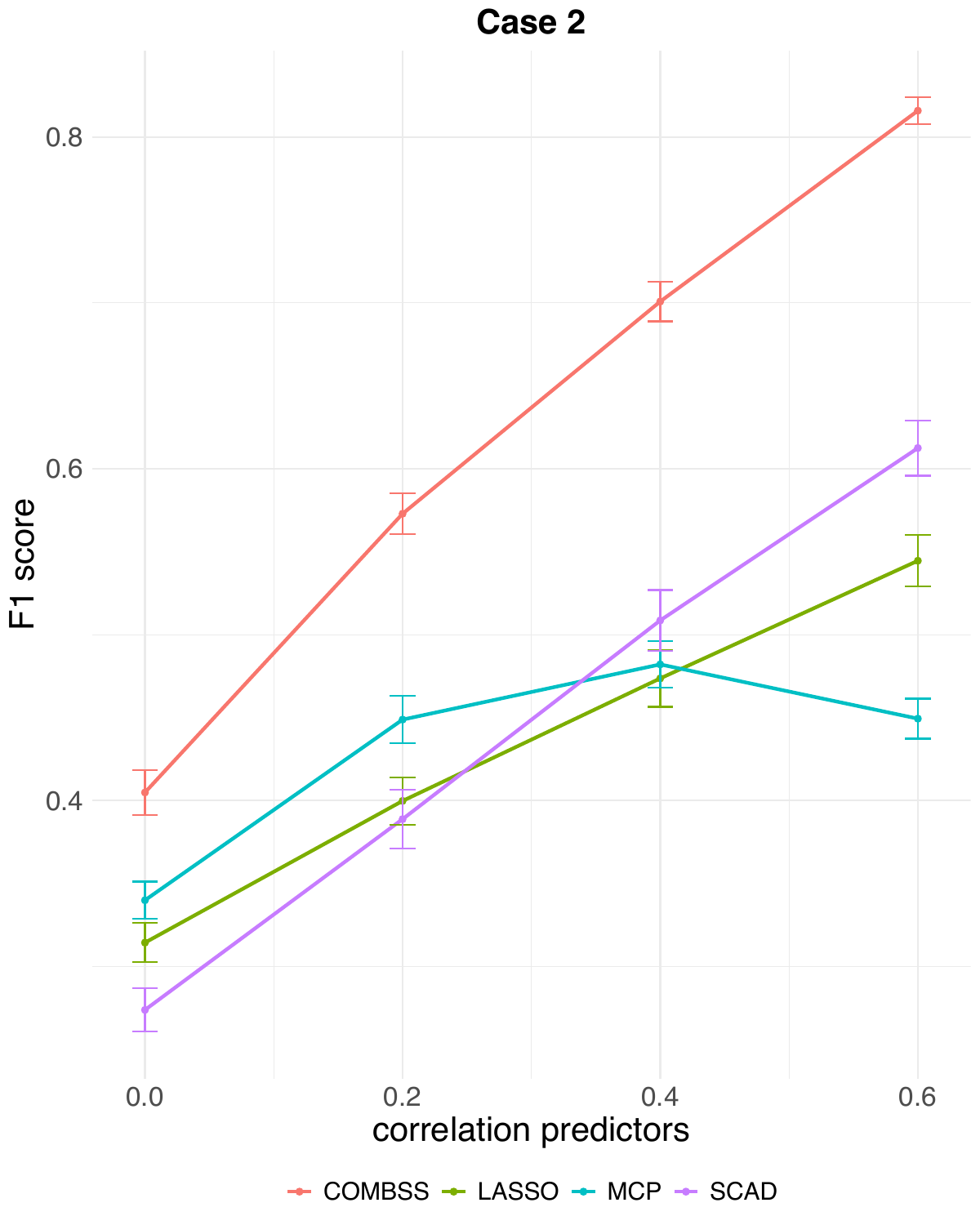}
    \caption{F1 score -- Case 2}
    \label{fig:high_case2_f1}
  \end{subfigure}
  \hfill
  \begin{subfigure}[b]{0.32\textwidth}
    \centering
    \includegraphics[width=\textwidth]{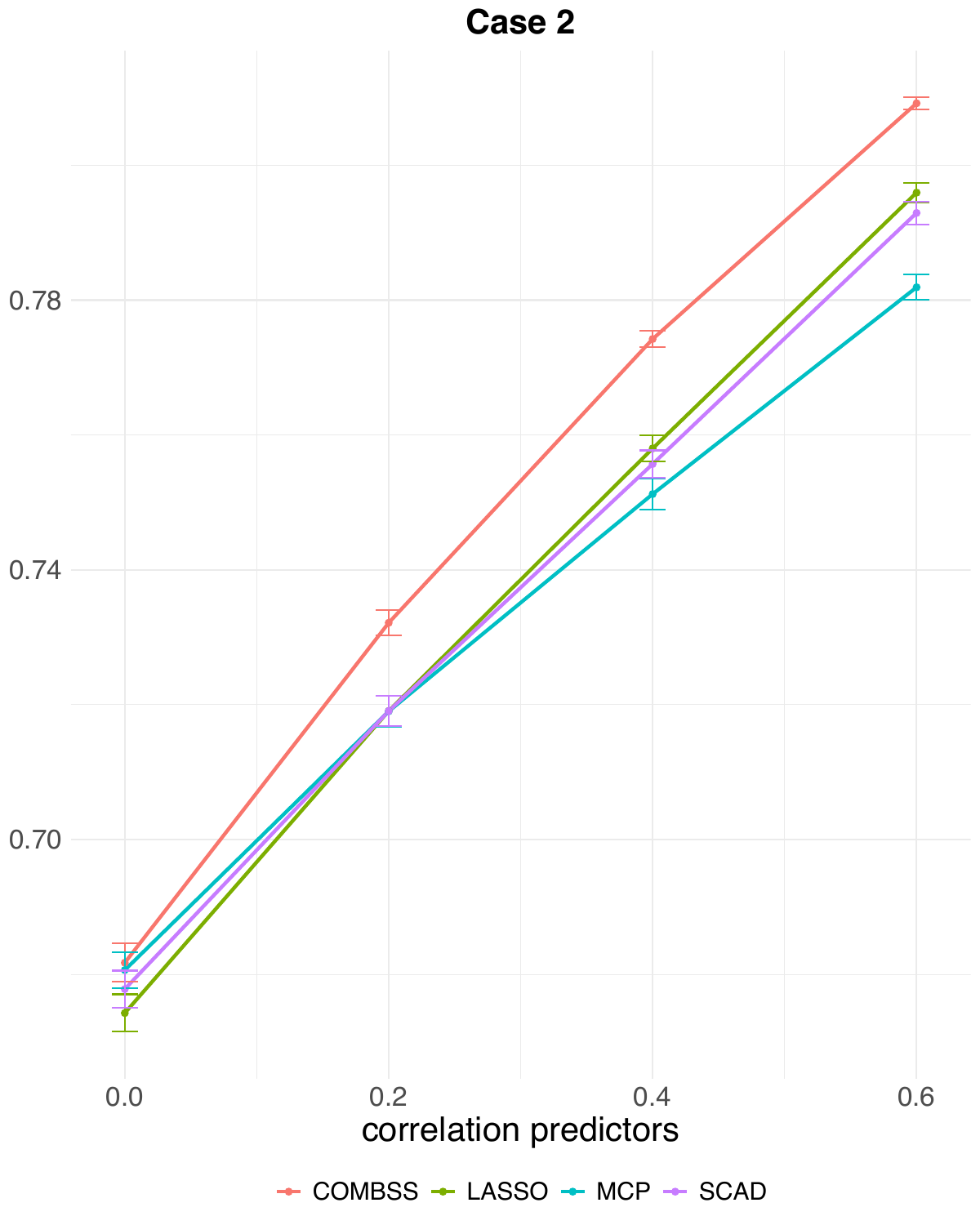}
    \caption{Prediction accuracy -- Case 2}
    \label{fig:high_case2_acc}
  \end{subfigure}
  \caption{Performance results in the high-dimensional setting
    ($n = 200$, $p = 1000$) for Case~1 (top row) and Case~2 (bottom
    row). Each panel displays the average over 50 replications as a
    function of predictor correlation $\rho \in \{0, 0.2, 0.4, 0.6\}$,
    with vertical bars denoting one standard error.}
  \label{fig:sim_high}
\end{figure}


\subsubsection{Running Time}

We record the wall-clock time for each replication of the simulation study and report the mean and standard deviation (in seconds) in Table~\ref{tbl:running_times}. The reported time corresponds to the total execution of Algorithm~\ref{alg:fw-homotopy-joint}, with $N = 25$, across model sizes $k = 1, \ldots, 20$; model evaluation steps are not included. For the low-dimensional setting ($p = 30$), COMBSS-GLM completes in fewer than 2~seconds across all designs and correlation levels. In the high-dimensional setting ($p = 1000$), runtimes range from approximately 3 to 5~seconds, demonstrating that the method scales moderately with the number of variables. All timings were obtained on an Intel Core i9-12900K with 32\,GB of RAM running Windows~10.
\begin{table}[htbp]
    \centering
    \caption{COMBSS-GLM wall-time in seconds: mean (std) over 50 replications for the simulated data with $n=200$.}
    \label{tab:wall_time}
    \begin{tabular}{cc rrrr}
    \toprule
    Design & $p$ & $\rho = 0$ & $\rho = 0.2$ & $\rho = 0.4$ & $\rho = 0.6$ \\
    \midrule
    \multirow{2}{*}{1} & 30 & 1.13 (0.03) & 1.24 (0.14) & 1.14 (0.02) & 1.25 (0.14) \\
      & 1000 & 3.54 (0.10) & 4.56 (1.43) & 3.72 (0.16) & 3.80 (0.12) \\
    \midrule
    \multirow{2}{*}{2} & 30 & 1.10 (0.04) & 1.20 (0.09) & 1.16 (0.09) & 1.24 (0.13) \\
      & 1000 & 3.74 (0.14) & 3.77 (0.11) & 3.75 (0.15) & 3.79 (0.08) \\
    \bottomrule
    \end{tabular}
    \label{tbl:running_times}
\end{table}

\section{Applications to Biomedical Data}
\label{sec:applications}

We now illustrate the COMBSS-GLM on two real biomedical datasets \citep{mccouch2016open, khan2001classification} that highlight
different aspects of the methodology.  The first application
(Section~\ref{sec:gwas}) applies the binary logistic formulation to a
high-dimensional genome-wide association study with $p = 158{,}210$ SNPs and
$n = 1{,}155$ samples.  The second application (Section~\ref{sec:khan}) applies the multinomial extension (check Section~S1 in the supplementary material) to a four-class cancer
classification problem with $p = 2{,}308$ genes and $n = 63$ samples.
Together, these two case studies demonstrate the generality of COMBSS-GLM across
different GLM families, numbers of response classes, and dimensional regimes.

\subsection{GWAS for Rice Grain Length Traits}
\label{sec:gwas}

To demonstrate the utility of efficient best subset selection for sparse generalized
linear models, we applied COMBSS-GLM to a high-dimensional dataset derived
from a publicly available genome-wide association study (GWAS) on rice.  This dataset
includes $1{,}155$ rice accessions genotyped with approximately $158{,}210$ single
nucleotide polymorphisms (SNPs) selected from the original High-Density Rice Array
(HDRA, 700k SNPs) \citep{mccouch2016open}.

We applied our proposed method after applying standard preprocessing procedures, detailed in \citet{Ullah2025}.  In summary, single nucleotide polymorphisms (SNP) loci were excluded if
their call rates were below $95\%$, their minor allele frequency (MAF) was less than
$5\%$, or if they deviated significantly from Hardy--Weinberg equilibrium.  Missing
genotype values for each SNP, which accounted for no more than $5\%$ of the values
for any single marker, were imputed using the mean of the available non-missing values
for that SNP.  After preprocessing, the resultant genotype data were scaled so that
each SNP had zero mean and unit variance, producing a design matrix of size $n = 1155$
and $p = 158{,}210$.  The phenotypic data consisted of continuous measurements of the average grain length of rice accessions, which exhibited a bimodal distribution with a
separation around $6$\,mm.  The response variable was then dichotomised into two categories according to whether the grain length exceeded $6$\,mm, allowing the resulting binary outcome to
serve as the dependent variable in a logistic regression setting.  Additionally, to account for
possible population stratification, the genotype data were preprocessed using the singular
value decomposition (SVD).  For further details on the data preparation
process, readers are referred to \citet{Ullah2025}.  We provide the final design
matrix alongside the corresponding binary outcome in the GitHub repository accompanying the paper.

The COMBSS-GLM methodology allowed for the identification of the best subset of predictors for model sizes, ranging from $k=1,\ldots,20$.  The best-subset inclusion
path for models of size $k = 1$ to $k = 10$ is shown in
Figure~\ref{fig:combined_figure}(a), while the full results for larger model sizes are
presented in the supplementary material.  Notably, the \textit{SNP-3.16732086}, a functional SNP
within the \textit{GS3} gene on chromosome~3 \citep{fan2006gs3,mao2010linking}, was
consistently selected across all model sizes.  This SNP has been previously reported
to be strongly associated with rice grain length and is known to play a crucial role in
regulating the trait \citep{mccouch2016open,Ullah2025}.  Additionally, a cluster of
highly correlated SNPs located on chromosome~3 was repeatedly selected across models
of varying sizes, a result that is consistent with existing findings.

\begin{figure}[t]
  \centering
  \begin{subfigure}[b]{0.48\textwidth}
    \centering
    \includegraphics[height=5cm]{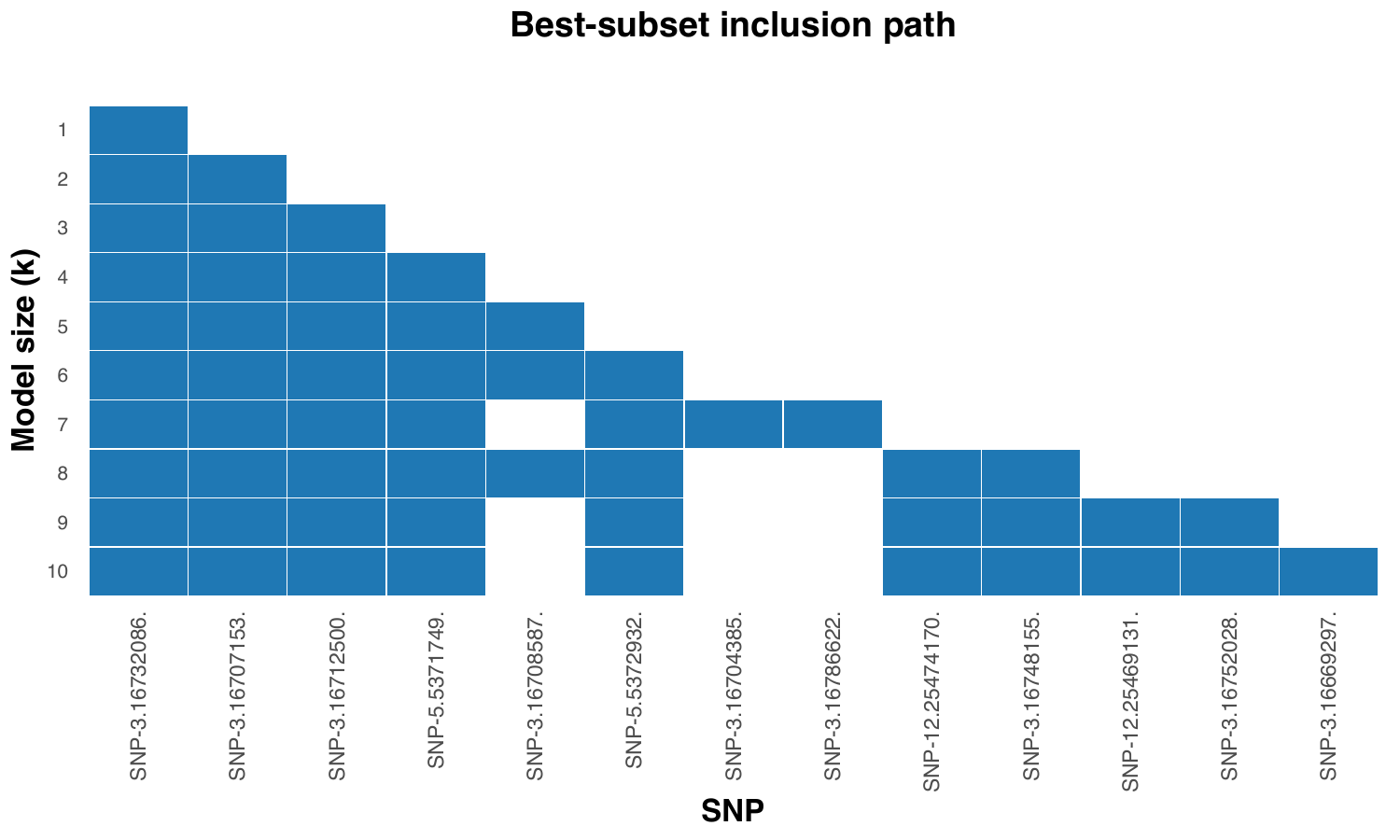}
    \caption{Best-subset inclusion path for models.}
    \label{fig:best_subset_inclusion_path}
  \end{subfigure}
  \hfill
  \begin{subfigure}[b]{0.48\textwidth}
    \centering
    \includegraphics[height=5cm]{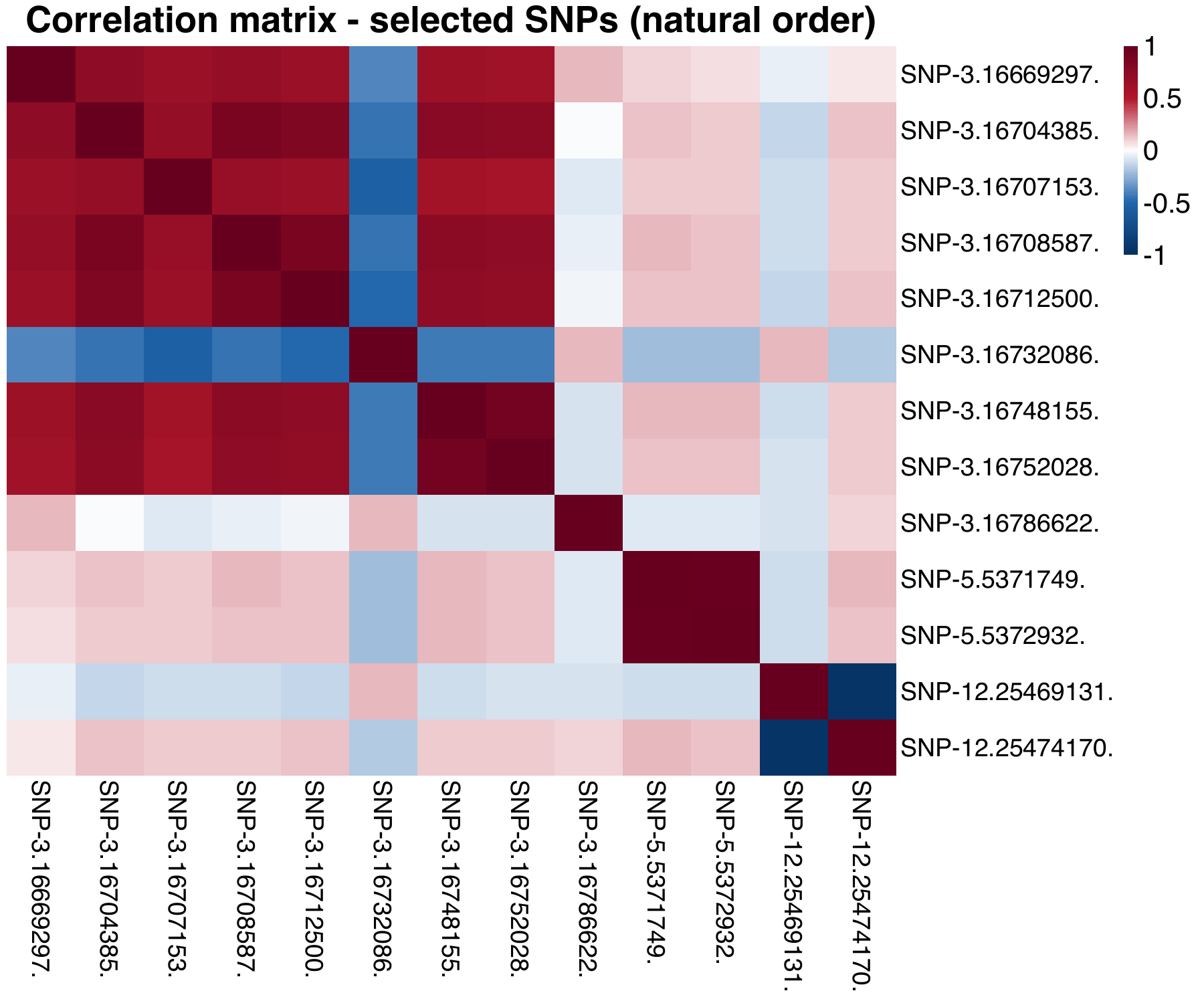}
    \caption{Correlation matrix of the selected SNPs.}
    \label{fig:correlation_matrix}
  \end{subfigure}
  \caption{(a) Best-subset inclusion path showing the selected SNPs for different model
    sizes $k = 1$ to $k = 10$.  (b) Correlation matrix for the selected SNPs
    showcasing relationships between the predictors.}
  \label{fig:combined_figure}
\end{figure}

In addition to chromosome~3, our analysis identified a SNP located on chromosome~5 as
a relevant predictor.  This SNP has also been identified in prior studies as associated
with grain size, including its mention as a potential locus associated with grain length
in \citet{li2011natural}.  Interestingly, SNPs on chromosome~12 appeared exclusively
in intermediate model sizes ($k = 8$, $9$, $10$), suggesting a potential role in
explaining grain length variation when considered alongside other loci.

Expanding the analysis to larger models ($k = 11$ to $k = 20$), as shown in
Figure~\ref{supplefig:combined_figure} in the supplementary material, we observed the
inclusion of SNPs on chromosome~4 within specific model sizes (e.g., $k = 11$--$14$
and $k = 18$--$20$).  Notably, these results have previously been reported in
\citet{mccouch2016open} and were also recently observed in \citet{Ullah2025}.  This
provides further support for a potential association of these loci with grain length
traits.

While the majority of selected SNPs were located on chromosome~3, consistent with the
dominance of this region, a very small number of loci were identified on chromosomes 1,
2, and~7 when larger model sizes were considered.  This aligns with the observation of
genome-wide quantitative trait loci (QTLs) for grain length and other related traits,
as documented in the literature \citep{huang2013genetic,zuo2014molecular}.

It is worth noting that the dichotomisation of the continuous grain length phenotype
into a binary trait, while simplifying the interpretation of our GLM-based model
selection method, may have introduced confounding effects.  Specifically, the
segmentation could obscure associations with other traits related to grain length or
amplify associations with markers that would otherwise exhibit weaker signals under a
continuous outcome.  Nonetheless, our method's robust performance in recovering
previously recognised loci, as well as identifying potential novel regions of interest
in a computationally efficient fashion, demonstrates its utility for high-dimensional
studies such as GWAS.

\subsection{Cancer Classification Using Gene Expression Data}
\label{sec:khan}

To further demonstrate the generality of the proposed methods, we appy the COMBSS-GLM using the multinomial likelihood(Section~S1 in the supplementary material) to a well-known cancer classification problem involving high-dimensional gene expression data.  This application illustrates the method's ability to perform best subset selection under a multinomial logistic model with $C > 2$ classes and $p \gg n$.

We consider the small round blue cell tumour (SRBCT) dataset of
\citet{khan2001classification}, which contains expression levels of $p = 2{,}308$ genes measured on tissue samples belonging to four distinct childhood tumour types:
Ewing sarcoma (EWS), Burkitt lymphoma (BL), neuroblastoma (NB), and rhabdomyosarcoma
(RMS).  The dataset is split into a training set of $n_{\mathrm{train}} = 63$ samples
and an independent test set of $n_{\mathrm{test}} = 20$ samples, as defined in the
original study.  This partition, which is preserved in the \textsf{R} package
\texttt{ISLR2} \citep{james2021introduction}, was used without modification throughout
our analysis.  The four-class structure ($C = 4$) and extreme dimensionality ratio
($p/n \approx 37$) make this dataset a challenging and widely studied benchmark for
high-dimensional classification and variable selection methods
\citep{khan2001classification,tibshirani2002diagnosis,zou2005regularization}.

We applied the multinomial formulation of COMBSS-GLM with $C = 4$ classes, given by~\eqref{eqn:h-general}:
\begin{equation*}
  h_{\delta,\lambda}(\boldsymbol{t}, \beta_0, B)
  = -\frac{1}{n}\,\ell(\beta_0, T_{\boldsymbol{t}} B;\, D)
    + \lambda\,\|B\|_F^2
    + \delta\,\|\Gamma_{\boldsymbol{t}} B\|_F^2,
\end{equation*}
with $B \in \mathbb{R}^{p \times (C-1)}$ denoting the coefficient matrix (one column
per non-baseline class) and $\beta_0 \in \mathbb{R}^{C-1}$ the intercept vector.
Algorithm~\ref{alg:fw-homotopy-joint} was applied to obtain subsets with sizes $k = 1, \ldots, 50$
 and with the ridge parameter set to $\lambda = 0$, so
that the homotopy curvature parameter~$\delta$ alone drives the relaxed variables toward binary
corners. The curvature schedule parameters $\delta_{\min}$, $\delta_{\max}$, and the
growth rate~$r$ were automatically calibrated from the data as described in
Section~\ref{sec:algorithm}, and a grid size of $N = 25$ was used. We use a learning rate $\alpha=0.01$ (default option).  Prior to
running the algorithm, the gene expression matrix was column-normalised to unit
$\ell_2$ norm as described in Section~\ref{sec:col-normalization}.  For each subset size~$k$, the selected genes were refit on the original training data using a ridge-penalised multinomial logistic model, with the ridge parameter chosen by cross-validation on the training set.  Classification performance was then evaluated on the held-out test set.  

As a benchmark, we applied the multinomial group Lasso using the \texttt{glmnet}
package in \textsf{R} \citep{friedman2010regularization} with
\texttt{family = "multinomial"} and \texttt{type.multinomial = "grouped"}, which
enforces a common sparsity pattern across the $C - 1 = 3$ class-specific coefficient
vectors.  The regularisation parameter $\lambda$ in the group Lasso was selected via cross-validation  with
misclassification error as the criterion.  We report three variants: the model at $\lambda_{\min}$ (the value minimising cross-validation error), the model at
$\lambda_{\text{1se}}$ (the largest $\lambda$ within one standard error of the
minimum), and the best model selected over the full $\lambda$ path by evaluating test accuracy directly.

Table~\ref{tab:khan_results} summarises the classification performance of COMBSS-GLM and
the multinomial group Lasso on the Khan test set.
\begin{table}[h]
  \centering
  \caption{Classification performance on the Khan SRBCT test set
    ($n_{\mathrm{test}} = 20$). ``Genes'' denotes the number of genes selected
    (out of $p = 2{,}308$). Test accuracy is the proportion of correctly classified
    test samples.}
  \label{tab:khan_results}
  \begin{tabular}{lcc}
    \toprule
    Method & \# genes selected & Test accuracy \\
    \midrule
    COMBSS-GLM ($k = 10$)                    & 5 & 0.85 \\
    COMBSS-GLM ($k = 13$)                    & 8 & 0.95 \\
    COMBSS-GLM ($k = 16$)                    & 12 & 1.00 \\
    Group Lasso ($\lambda_{\text{1se}}$) & 28 & 0.95 \\
    Group Lasso ($\lambda_{\min}$)       & 30 & 0.95 \\
    Group Lasso (best $\lambda$)         & 35 & 1.00 \\
    \bottomrule \\
  \end{tabular}
\end{table}

COMBSS-GLM achieved perfect classification accuracy ($20/20$) on the independent test set
using only $12$ genes.  By contrast, the multinomial group Lasso required $35$ genes
to reach the same performance, and its standard cross-validated models
($\lambda_{\min}$ and $\lambda_{\text{1se}}$) selected $30$ and $28$ genes
respectively while achieving only $95\%$ test accuracy.  Already at $k = 8$ genes, COMBSS-GLM matched this $95\%$ accuracy using
fewer than a third of the number of genes selected by either
cross-validated Lasso model.  This represents a reduction of more than
$70\%$ in model size relative to the group Lasso for equivalent
predictive performance, highlighting the ability of best subset
selection to identify highly parsimonious yet accurate multiclass
classifiers in the $p \gg n$ regime.

Figure~\ref{fig:khan_path} displays the best-subset inclusion path for model sizes
$k = 1$ to $k = 20$, together with the test classification accuracy for each subset
size.  Gene~$1954$ is the first to enter the model (at $k = 1$) and persists
throughout, joined by Gene~$1955$ at $k = 2$ and Gene~$246$ at $k = 3$.  A core set
of $8$ genes (Genes~$187$, $246$, $509$, $1389$, $1645$, $1954$, $1955$, and $2050$)
is established by $k = 8$, at which point test accuracy reaches $95\%$.  Accuracy
remains at $95\%$ as further genes are added and reaches perfect classification
($100\%$, highlighted in red in Figure~\ref{fig:khan_path}) at $k = 12$ with only
$12$ out of $2{,}308$ genes.  Table~\ref{tab:khan_genes} in
Appendix~\ref{app:khan_genes} provides the gene annotations for these $12$ genes,
obtained from the \texttt{srbct} dataset in the \textsf{R} package \texttt{mixOmics}
\citep{rohart2017mixomics}.

\begin{figure}[t]
  \centering
  \includegraphics[width=0.85\textwidth]{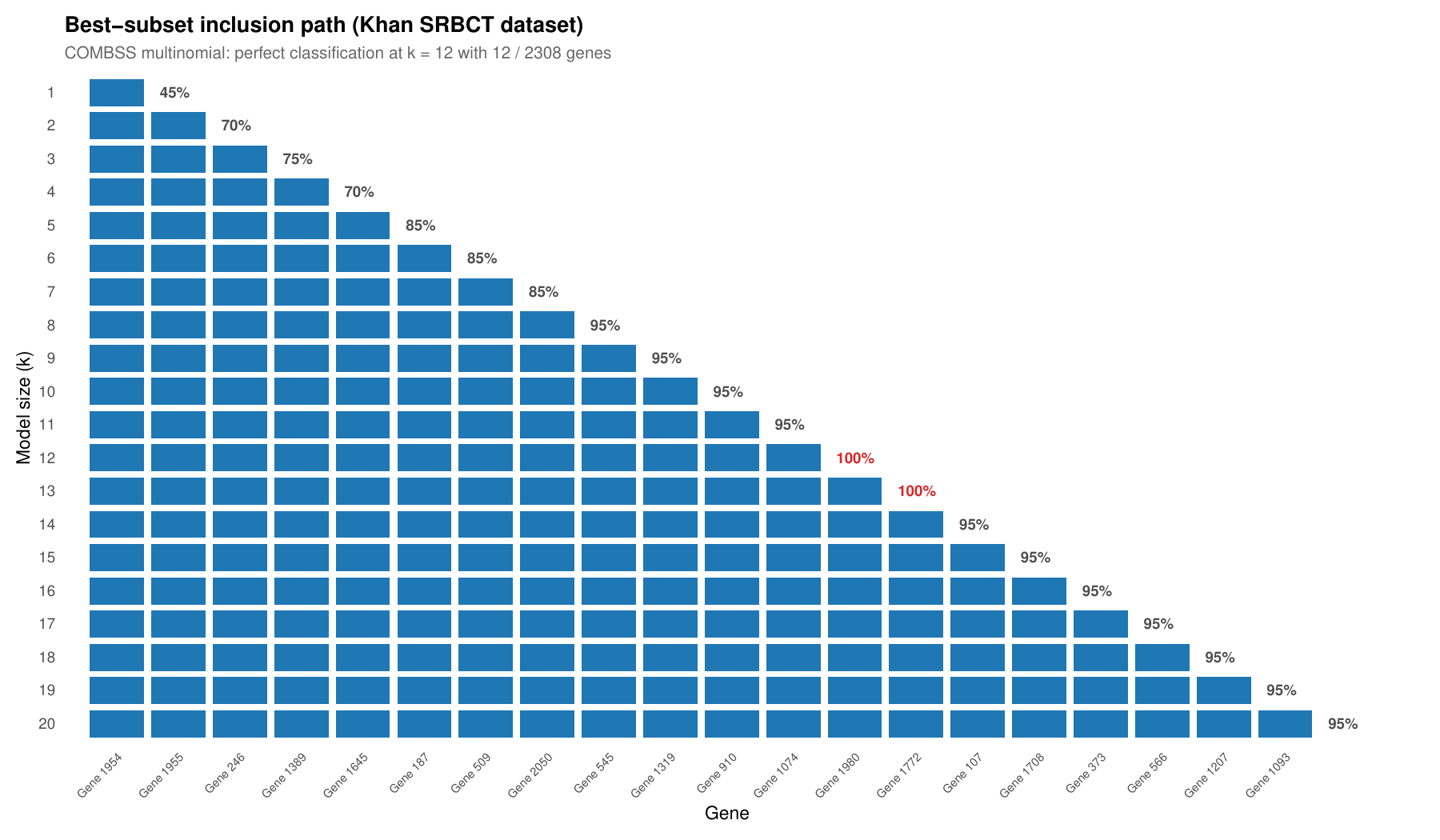}
  \caption{Best-subset inclusion path for the Khan SRBCT gene expression dataset.
    Each row corresponds to a model size $k = 1, \ldots, 20$, and each filled cell
    indicates that the corresponding gene is included in the best subset of that size.
    The test classification accuracy is displayed at the end of each row, with perfect
    accuracy ($100\%$) highlighted in red.  The path exhibits a monotone nesting
    property: once a gene enters the model, it remains selected at all larger model
    sizes.}
  \label{fig:khan_path}
\end{figure}

These results are consistent with and complementary to previous analyses of the Khan
dataset.  \citet{khan2001classification} employed artificial neural networks and
identified a set of $96$ genes achieving perfect classification, while
\citet{tibshirani2002diagnosis} applied the nearest shrunken centroid method (PAM) and
reduced this to $43$ genes with the same accuracy.  The proposed COMBSS-GLM method
achieves perfect classification with substantially fewer genes ($k = 12$),
demonstrating the advantage of direct best subset selection under a multinomial
logistic model over penalisation-based approaches that encourage but do not enforce
exact sparsity.

\section{Conclusion}
\label{sec:conclusion}

We have presented a continuous optimisation method for best subset
selection in generalised linear models, extending the COMBSS approach
of \citet{moka2024combss} from linear regression to the broader GLM
family.  The method replaces the combinatorial search over $\binom{p}{k}$
candidate subsets with a smooth optimisation on a simplex slice, solved
by a Frank--Wolfe homotopy that requires only one penalised GLM fit per
iteration.  We established sufficient conditions under which the relaxed
objective becomes concave in the selection weights, guaranteeing that
global minimisers correspond to valid subsets, and derived explicit
thresholds for logistic and multinomial models.

Simulations showed that the method consistently matches or outperforms
the Lasso, SCAD, and MCP in sparse recovery and
predictive accuracy, particularly when $p \gg n$.  In a rice GWAS with
$p = 158{,}210$ SNPs, the method recovered known trait-associated loci
within minutes on standard hardware.  In the Khan SRBCT problem
($C = 4$, $p = 2{,}308$, $n = 63$), COMBSS-GLM achieved perfect test
classification with only $12$ genes, more than $65\%$ fewer than the
multinomial group Lasso.

Natural extensions include other exponential-family members (e.g.\
Poisson, negative binomial), systematic comparison with mixed-integer
optimisation solvers \citep{BKM16}, and incorporation of
group-structured sparsity (in the same spirit as \cite{mathur2024group}) for applications such as pathway-based
genomic analyses.

\section*{Data and Code Availability}
Open-source implementations of the proposed method in both R and Python,
together with reproducible scripts for all simulation studies and
biomedical applications, are publicly available in R at
\url{https://github.com/benoit-liquet/COMBSS-GLM-R} and in Python at
\url{https://github.com/anantmathur44/COMBSS-GLM-Python}.



\bibliographystyle{plainnat}
\bibliography{Ref}

\clearpage
\appendix

\renewcommand{\thesection}{S\arabic{section}}
\renewcommand{\thetable}{S\arabic{table}}
\renewcommand{\thefigure}{S\arabic{figure}}
\renewcommand{\theequation}{S\arabic{equation}}
\setcounter{section}{0}
\setcounter{table}{0}
\setcounter{figure}{0}
\setcounter{equation}{0}

\section*{Supplementary Material}
\label{sec:supplementary}
\addcontentsline{toc}{section}{Supplementary Material}

\section{Multinomial Regression}
\label{app:multinomial}

We briefly specialize the Boolean-relaxation framework to multinomial
regression with $C$ classes, using class $C$ as the baseline.
Let $\cD = (X, \v y)$, where $X \in \R^{n \times p}$ is the design matrix
(excluding the intercept) and $y_i \in \{1, \ldots, C\}$.
We parameterise the model by an intercept vector
$\bbt_0 \in \R^{C-1}$ and a coefficient matrix $B \in \R^{p \times (C-1)}$
(column $c$ corresponds to class $c$).
For $i \in [n]$ and $c \in [C-1]$, define the scores
\[
  \eta_{ic}(\bbt_0, B) := \beta_{0,c} + x_i^\top B_{:c},
\]
and probabilities
\[
  \pi_{ic}(\bbt_0, B) :=
    \frac{\exp(\eta_{ic}(\bbt_0,B))}{1+\sum_{d=1}^{C-1}\exp(\eta_{id}(\bbt_0,B))}
    \quad (c \in [C-1]),
\]
and
\[
  \pi_{iC}(\bbt_0, B) :=
    \frac{1}{1+\sum_{d=1}^{C-1}\exp(\eta_{id}(\bbt_0,B))}.
\]
Let $Y \in \R^{n \times (C-1)}$ be the one-hot encoding,
$Y_{ic} = \ind(y_i = c)$ for $c \in [C-1]$.
The log-likelihood $\ell(\bbt_0, B;\, \cD)$ is
equal to 
\[
  \sum_{i=1}^n \left(
    \sum_{c=1}^{C-1} Y_{ic}\,\eta_{ic}(\bbt_0, B)
    - \log\Bigl(1 + \sum_{d=1}^{C-1}\exp(\eta_{id}(\bbt_0, B))\Bigr)
  \right).
\]

\paragraph{Boolean relaxation with mandatory features.}
Assume $m$ mandatory features (with indices $1, \ldots, m$ without loss
generality).  As in the binary case, for $\v t \in [0,1]^{p-m}$ define the
diagonal matrices
\[
  (T_{\v t})_{jj} =
  \begin{cases}
    1, & j \le m,\\
    t_{j-m}, & j > m,
  \end{cases}
  \qquad\text{and}\qquad
  \Gamma_{\v t} = \sqrt{I - T_{\v t}^2}.
\]
Given $\delta > 0$ and $\lambda \ge 0$, define the penalised objective
\begin{equation}
  \label{eqn:h-general}
  h_{\delta,\lambda}(\v t, \bbt_0, B)
  =
  -\frac{1}{n}\ell(\bbt_0, T_{\v t} B;\, \cD)
  + \lambda\|B\|_F^2
  + \delta\|\Gamma_{\v t} B\|_F^2,
\end{equation}
and the value function
\[
  f_{\delta,\lambda}(\v t) := \min_{\bbt_0, B}\, h_{\delta,\lambda}(\v t, \bbt_0, B).
\]
The relaxed subset-selection problem is
$\min_{\v t \in \cT_k} f_{\delta,\lambda}(\v t)$ with
$\cT_k = \{\v t \in [0,1]^{p-m} : \v 1^\top \v t = k\}$.

\paragraph{Inner solves and gradients on the interior.}
For interior $\v t \in (0,1)^{p-m}$, it is convenient to move $\v t$ entirely into a quadratic penalty.
Let $\v\xi_0 := \bbt_0$ and $\Xi := T_{\v t} B$ (so $B = T_{\v t}^{-1}\Xi$)
and define
\begin{align*}
  \phi_{\delta,\lambda}(\v t, \v\xi_0, \Xi)
  &:=
  -\frac{1}{n}\ell(\v\xi_0, \Xi;\, \cD)
  + \lambda\|T_{\v t}^{-1}\Xi\|_F^2
  + \delta\|\Gamma_{\v t} T_{\v t}^{-1}\Xi\|_F^2\\
  &=
  -\frac{1}{n}\ell(\v\xi_0, \Xi;\, \cD)
  + \Bigl\|(\lambda I + \delta\,\Gamma_{\v t}^2)^{1/2} T_{\v t}^{-1}\Xi\Bigr\|_F^2.
\end{align*}
Then $f_{\delta,\lambda}(\v t) = \min_{\v\xi_0, \Xi}\,\phi_{\delta,\lambda}(\v t, \v\xi_0, \Xi)$,
and by Danskin's theorem, for any minimizer $(\wt{\v\xi}_0, \wt\Xi)$,
\[
  \nabla f_{\delta,\lambda}(\v t)
  = \nabla_{\v t}\,\phi_{\delta,\lambda}(\v t, \wt{\v\xi}_0, \wt\Xi),
\]
or a valid subgradient if the minimizer is not unique.
Since the likelihood term does not depend on $\v t$, we differentiate only
the penalty.  Using $T_{\v t}$ and $\Gamma_{\v t}$,
\begin{align}
  \Bigl\|(\lambda I + \delta\,\Gamma_{\v t}^2)^{1/2} T_{\v t}^{-1}\Xi\Bigr\|_F^2
  &=
  \lambda\sum_{j=1}^{m}\|\Xi_{j:}\|_2^2\\
  &\hspace{5mm} + \sum_{j=1}^{p-m}\left(\frac{\lambda+\delta}{t_j^{2}}-\delta\right)
    \|\Xi_{m+j,:}\|_2^2,
  \label{eqn:penalty-expansion}
\end{align}
and hence for $j = 1, \ldots, p-m$,
\[
  \frac{\partial}{\partial t_j}\,\phi_{\delta,\lambda}(\v t, \v\xi_0, \Xi)
  =
  -2(\lambda+\delta)\,\frac{\|\Xi_{m+j,:}\|_2^2}{t_j^{3}},
\]
and
\[
  \nabla_{\v t}\,\phi_{\delta,\lambda}(\v t, \v\xi_0, \Xi)
  =
  -2(\lambda+\delta)
  \begin{pmatrix}
    \|\Xi_{m+1,:}\|_2^2/t_1^3\\
    \vdots\\
    \|\Xi_{p,:}\|_2^2/t_{p-m}^3
  \end{pmatrix}.
\]

\paragraph{Using solvers without coefficient-wise penalties.}
The expansion~\eqref{eqn:penalty-expansion} shows that, for fixed interior
$\v t$, the penalty is a row-weighted ridge term.
If coefficient-wise weights are unavailable, define
\[
  \omega_j(\v t) :=
  \begin{cases}
    \lambda, & j = 1, \ldots, m,\\[2pt]
    \displaystyle\frac{\lambda+\delta}{t_{j-m}^{2}} - \delta,
      & j = m+1, \ldots, p,
  \end{cases}
\]
reparameterise $\Theta := \diag(\sqrt{\v\omega(\v t)})\,\Xi$, and scale the
design as $\widetilde X_{\v t} := X\,\diag(\v\omega(\v t)^{-1/2})$.
Then $X\Xi = \widetilde X_{\v t}\Theta$ and
$\sum_{j=1}^p \omega_j(\v t)\|\Xi_{j,:}\|_2^2 = \|\Theta\|_F^2$, so the
inner problem reduces to a standard ridge-penalised multinomial fit on
$(\widetilde X_{\v t}, \v y)$ with a uniform Frobenius penalty, followed
by back-transformation $\Xi = \diag(\v\omega(\v t)^{-1/2})\,\Theta$.

\section{Proofs}
\label{app:conv-proofs}
\begin{proof}[Proof of Proposition~\ref{prop:monotone-continuous-delta}]
Property~(i) follows because $(1-t_j^2)\ge 0$ for all
$\v t\in[0,1]^{p-m}$, so $h_{\delta_2,\lambda}\ge h_{\delta_1,\lambda}$
pointwise in $(\beta_0,\bbt)$, and taking the infimum preserves the
inequality.
Property~(ii) is a direct consequence of Berge's maximum
theorem~\citep{sundaram1996first}: $h_{\delta,\lambda}(\v t,\beta_0,\bbt)$
is jointly continuous in $(\delta,\beta_0,\bbt)$, and for every
$\delta>0$ the inner objective is coercive in $(\beta_0,\bbt)$
(owing to the ridge term $\delta\|\Gamma_{\v t}\bbt\|_2^2$),
so the minimising correspondence is compact-valued and continuous,
and Berge's theorem implies continuity of the value function
$f_{\delta,\lambda}(\v t)$ in $\delta$.
\end{proof}

\begin{proof}[Proof of Theorem~\ref{prop:concavity-threshold-unified}]
The stated bound implies
\[
  X_u^\top H_\eta(\eta) X_u \preceq 2\delta I
  \qquad\text{for all }\eta.
\]
Hence
\[
  D_{\bbt_u} X_u^\top H_\eta(\eta) X_u D_{\bbt_u} \preceq 2\delta D_{\bbt_u}^2.
\]
Using
\[
  \nabla_{\v t\v t}^2\, h_{\delta,\lambda}(\v t, \beta_0, \bbt)
  =
  D_{\bbt_u} X_u^\top H_\eta(\eta) X_u D_{\bbt_u}
  - 2\delta D_{\bbt_u}^2,
\]
we obtain $\nabla_{\v t\v t}^2\, h_{\delta,\lambda} \preceq 0$, so
$\v t \mapsto h_{\delta,\lambda}(\v t, \beta_0, \bbt)$ is concave for each
fixed $(\beta_0, \bbt)$.
Now
\[
  f_{\delta,\lambda}(\v t) = \inf_{\beta_0, \bbt}\, h_{\delta,\lambda}(\v t, \beta_0, \bbt)
\]
is the pointwise infimum of concave functions, hence concave.
Finally, minimizing a concave function over the compact polytope $\cT_k$ attains a minimizer at an extreme point, and the extreme points of $\cT_k$
are exactly $\cS_k$.
\end{proof}


\clearpage
\section{Additional Simulations}
\label{app:add-sims}

\begin{figure}[H]
  \centering
  \begin{subfigure}[b]{0.32\textwidth}
    \centering
    \includegraphics[width=\textwidth]{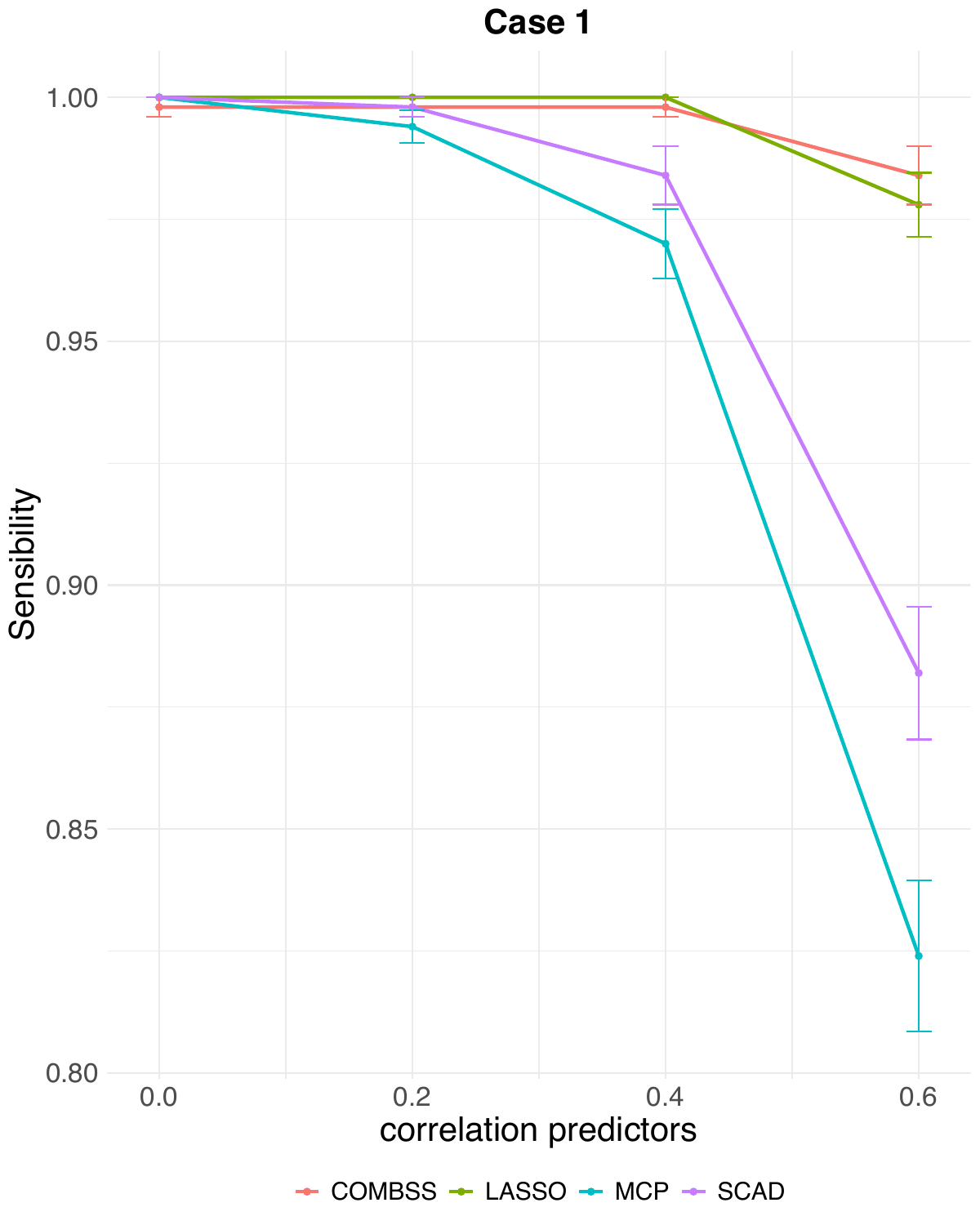}
    \caption{Sensitivity -- Case~1}
    \label{fig:low_case1_sens}
  \end{subfigure}
  \hfill
  \begin{subfigure}[b]{0.32\textwidth}
    \centering
    \includegraphics[width=\textwidth]{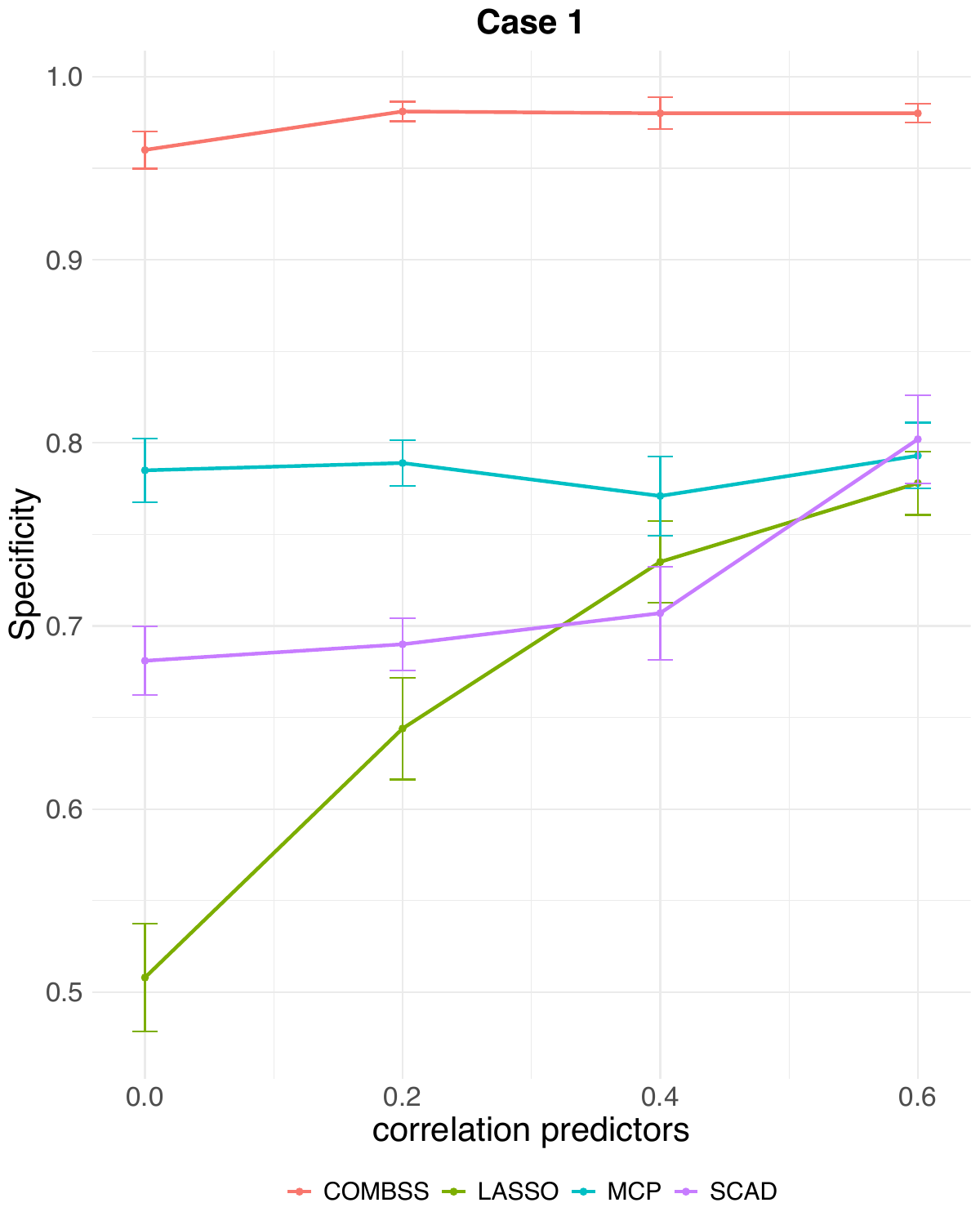}
    \caption{Specificity -- Case~1}
    \label{fig:low_case1_spec}
  \end{subfigure}
  \hfill
  \begin{subfigure}[b]{0.32\textwidth}
    \centering
    \includegraphics[width=\textwidth]{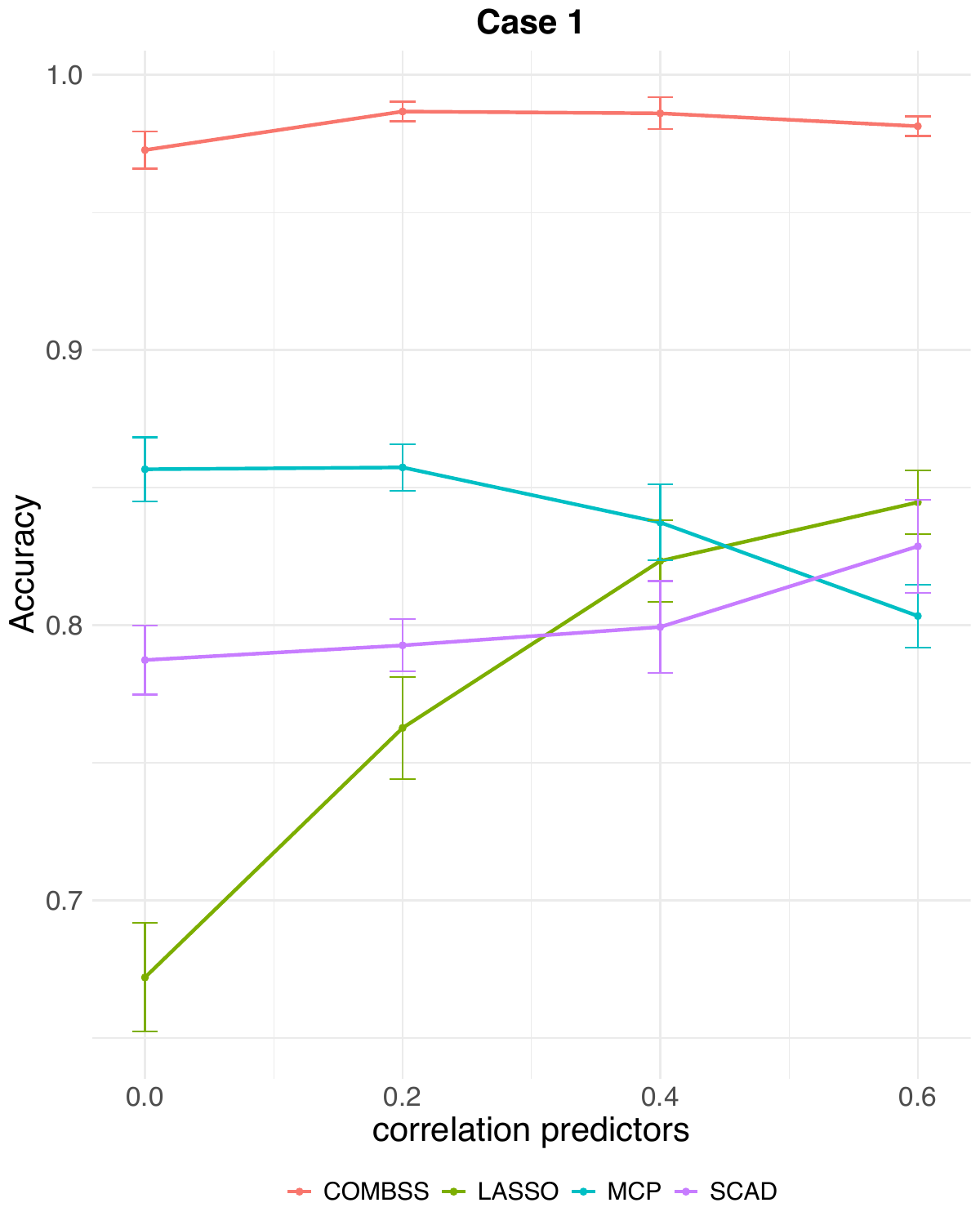}
    \caption{Selection accuracy -- Case~1}
    \label{fig:low_case1_selacc}
  \end{subfigure}

  \vspace{0.4cm}

  \begin{subfigure}[b]{0.32\textwidth}
    \centering
    \includegraphics[width=\textwidth]{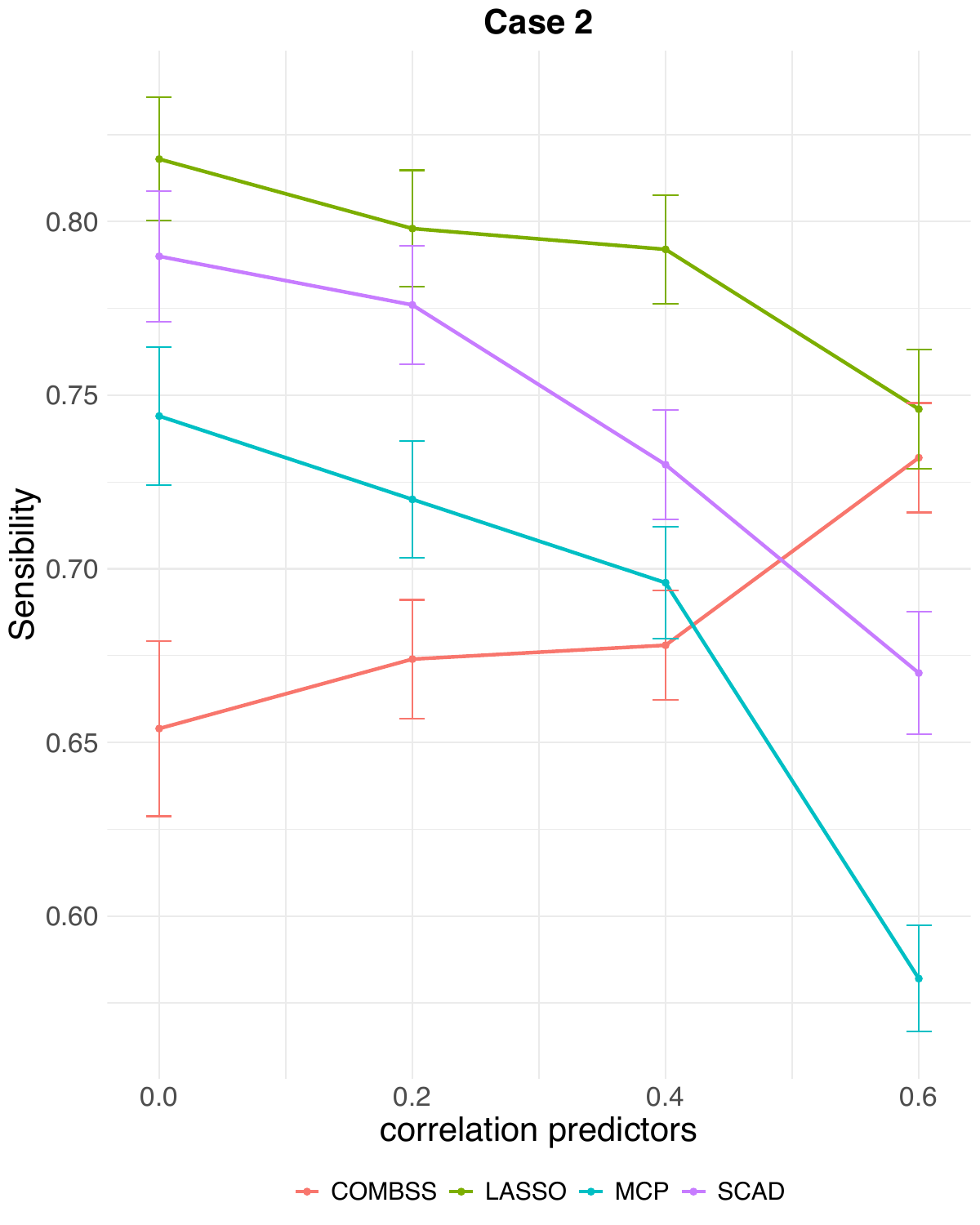}
    \caption{Sensitivity -- Case~2}
    \label{fig:low_case2_sens}
  \end{subfigure}
  \hfill
  \begin{subfigure}[b]{0.32\textwidth}
    \centering
    \includegraphics[width=\textwidth]{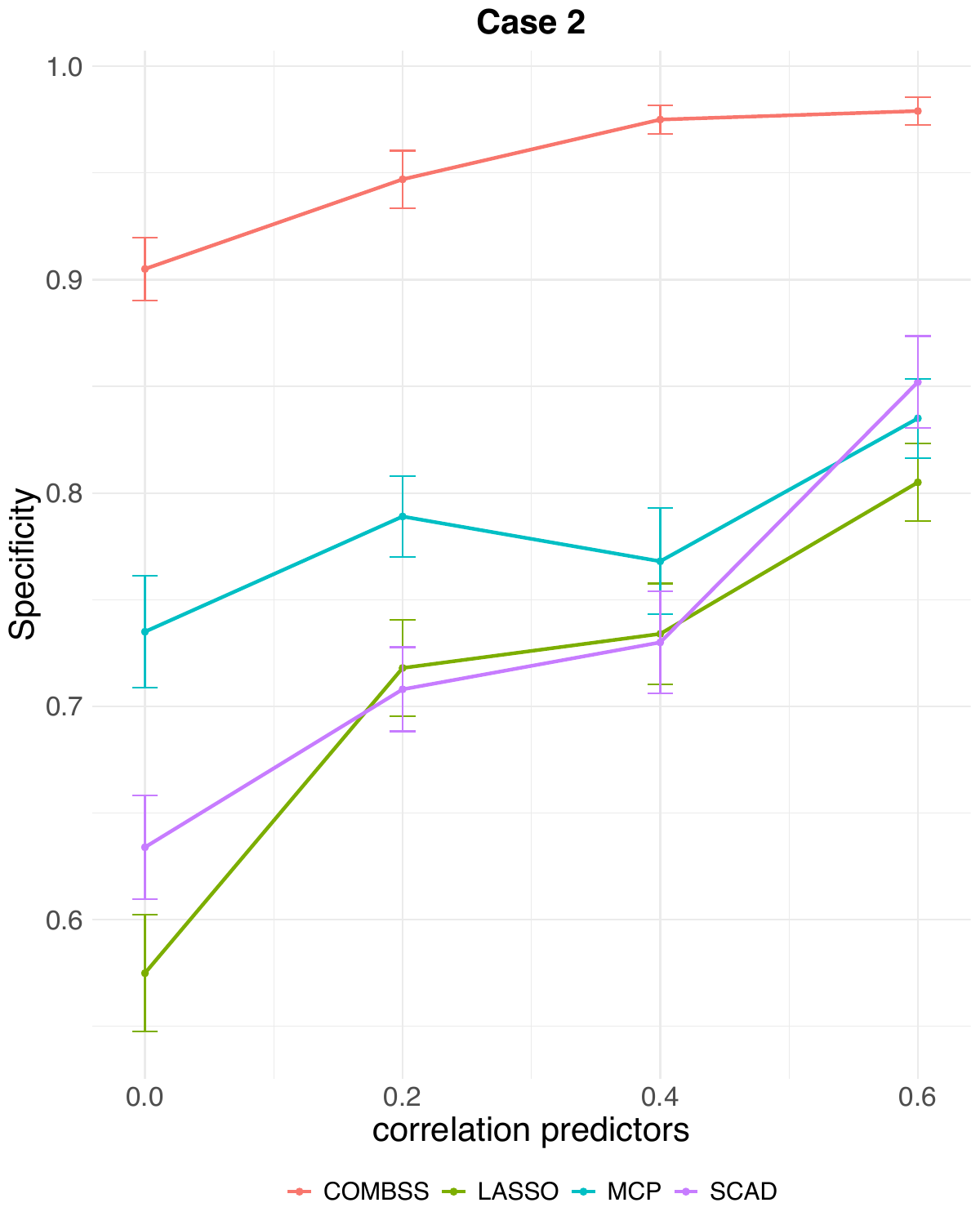}
    \caption{Specificity -- Case~2}
    \label{fig:low_case2_spec}
  \end{subfigure}
  \hfill
  \begin{subfigure}[b]{0.32\textwidth}
    \centering
    \includegraphics[width=\textwidth]{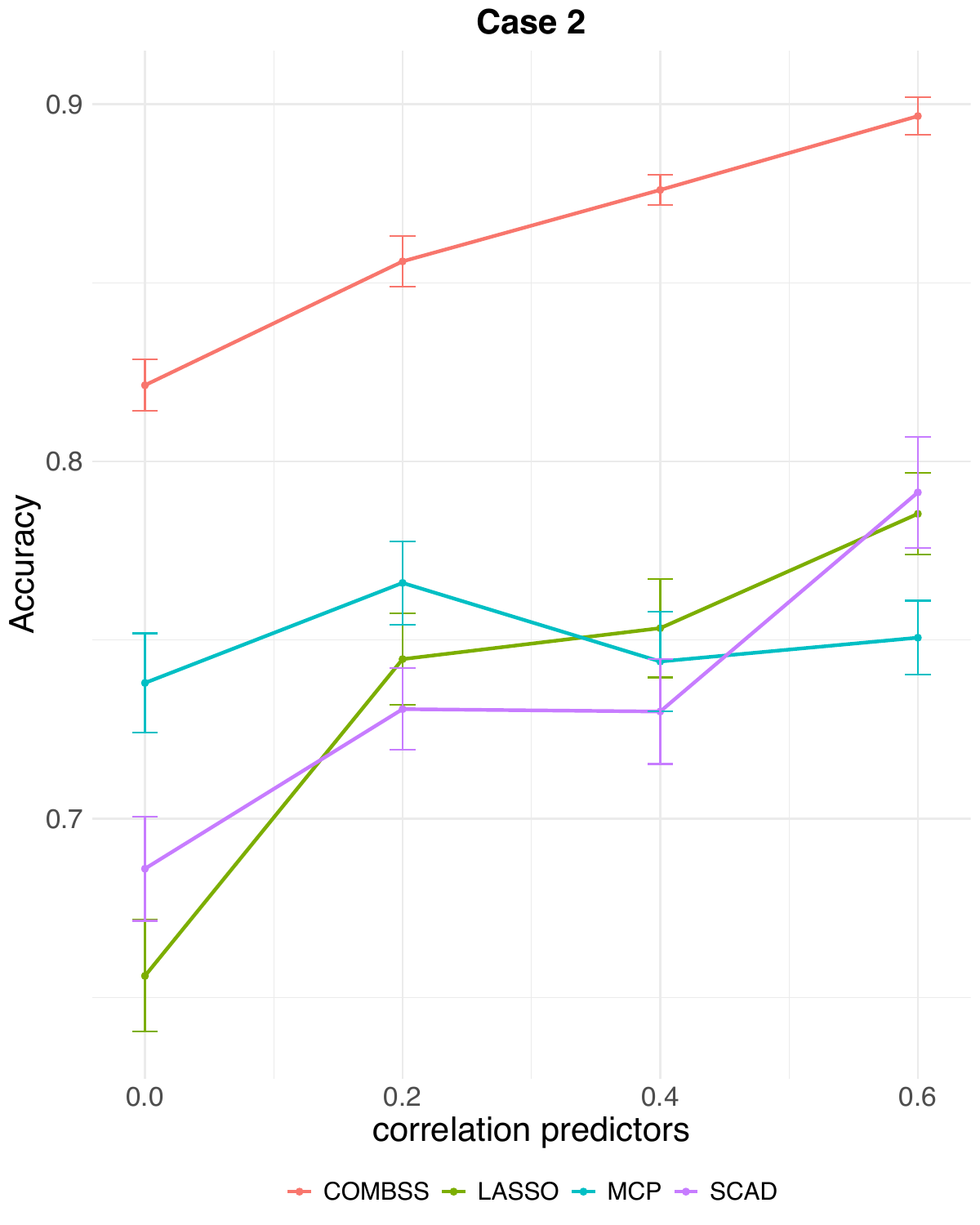}
    \caption{Selection accuracy -- Case~2}
    \label{fig:low_case2_selacc}
  \end{subfigure}
  \caption{Supplementary performance results in the low-dimensional setting
    ($n = 200$, $p = 30$) for Case~1 (top row) and Case~2 (bottom row).
    Each panel displays the average sensitivity, specificity, and selection
    accuracy over 50 replications as a function of the predictor correlation
    $\rho \in \{0, 0.2, 0.4, 0.6\}$, with vertical bars denoting one
    standard error.}
  \label{fig:supp_sim_low}
\end{figure}

\begin{figure}[H]
  \centering
  \begin{subfigure}[b]{0.32\textwidth}
    \centering
    \includegraphics[width=\textwidth]{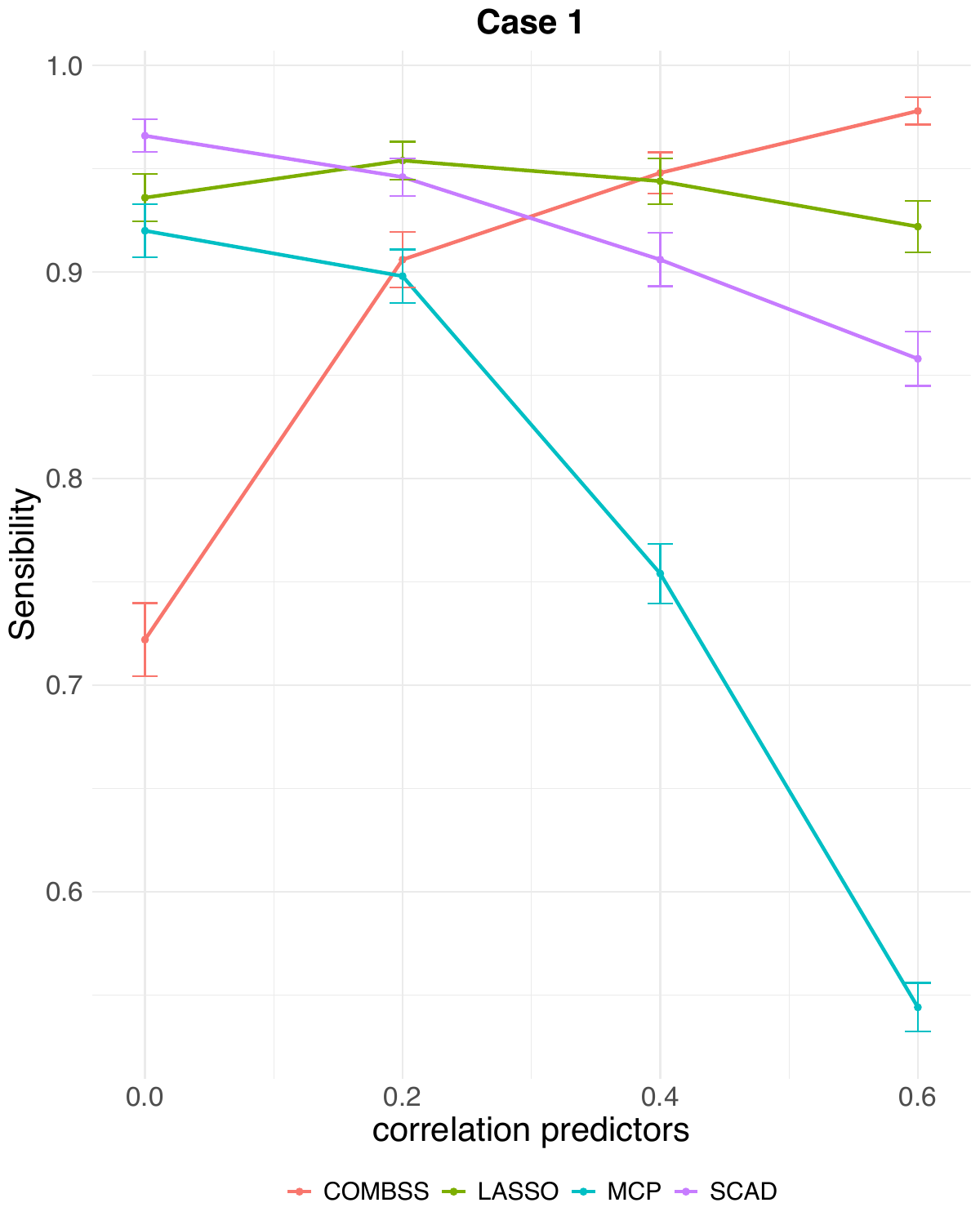}
    \caption{Sensitivity -- Case~1}
    \label{fig:high_case1_sens}
  \end{subfigure}
  \hfill
  \begin{subfigure}[b]{0.32\textwidth}
    \centering
    \includegraphics[width=\textwidth]{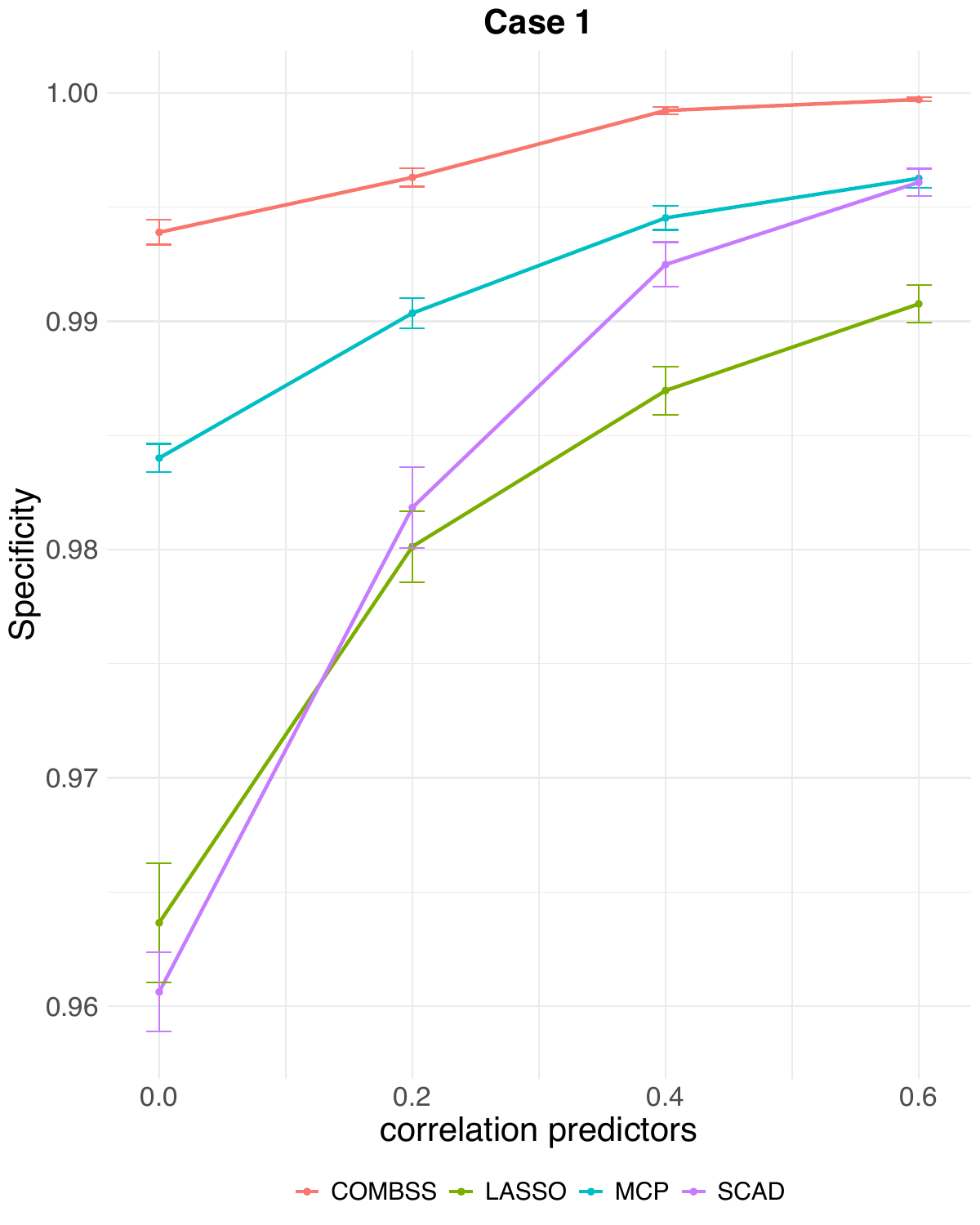}
    \caption{Specificity -- Case~1}
    \label{fig:high_case1_spec}
  \end{subfigure}
  \hfill
  \begin{subfigure}[b]{0.32\textwidth}
    \centering
    \includegraphics[width=\textwidth]{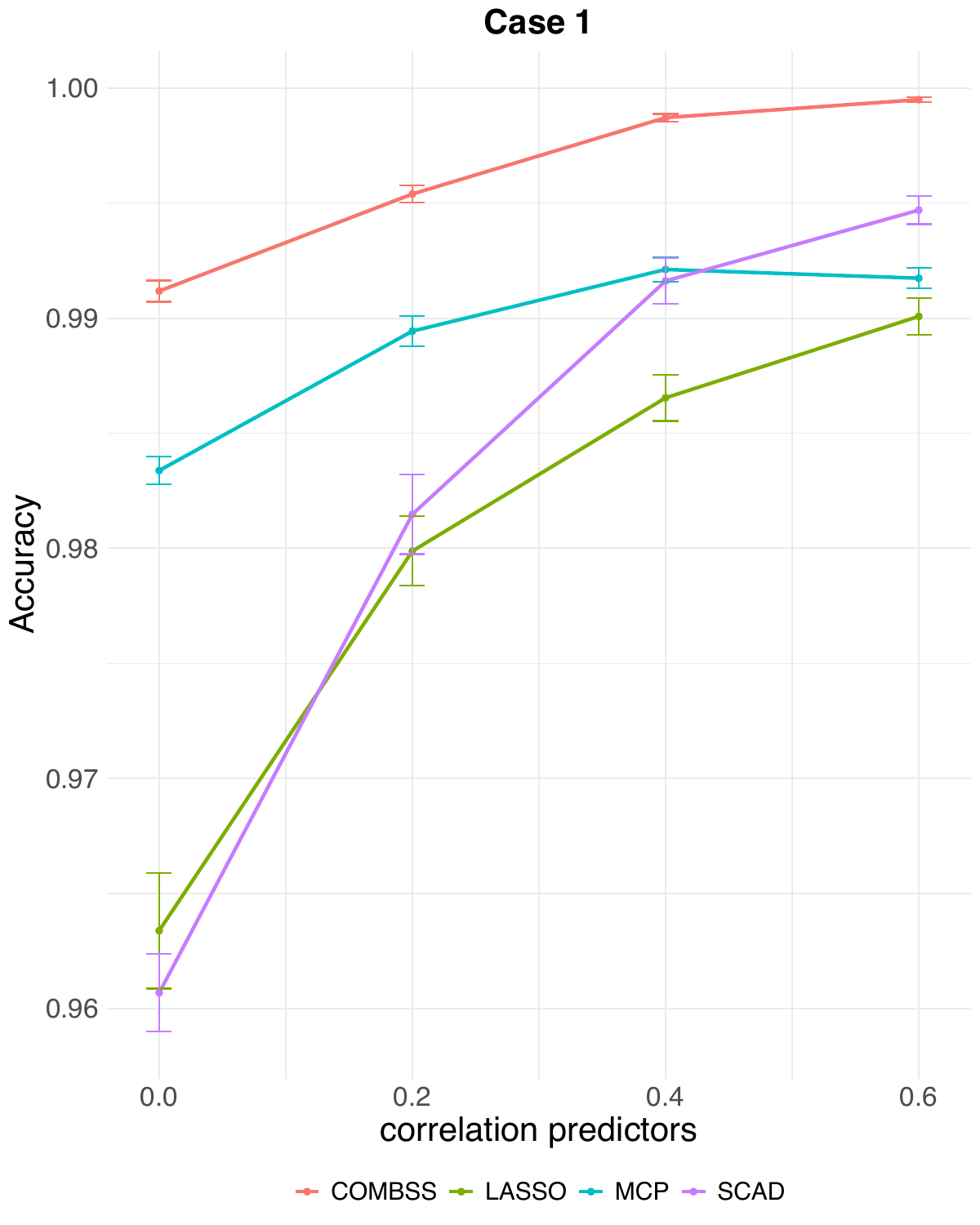}
    \caption{Selection accuracy -- Case~1}
    \label{fig:high_case1_selacc}
  \end{subfigure}

  \vspace{0.4cm}

  \begin{subfigure}[b]{0.32\textwidth}
    \centering
    \includegraphics[width=\textwidth]{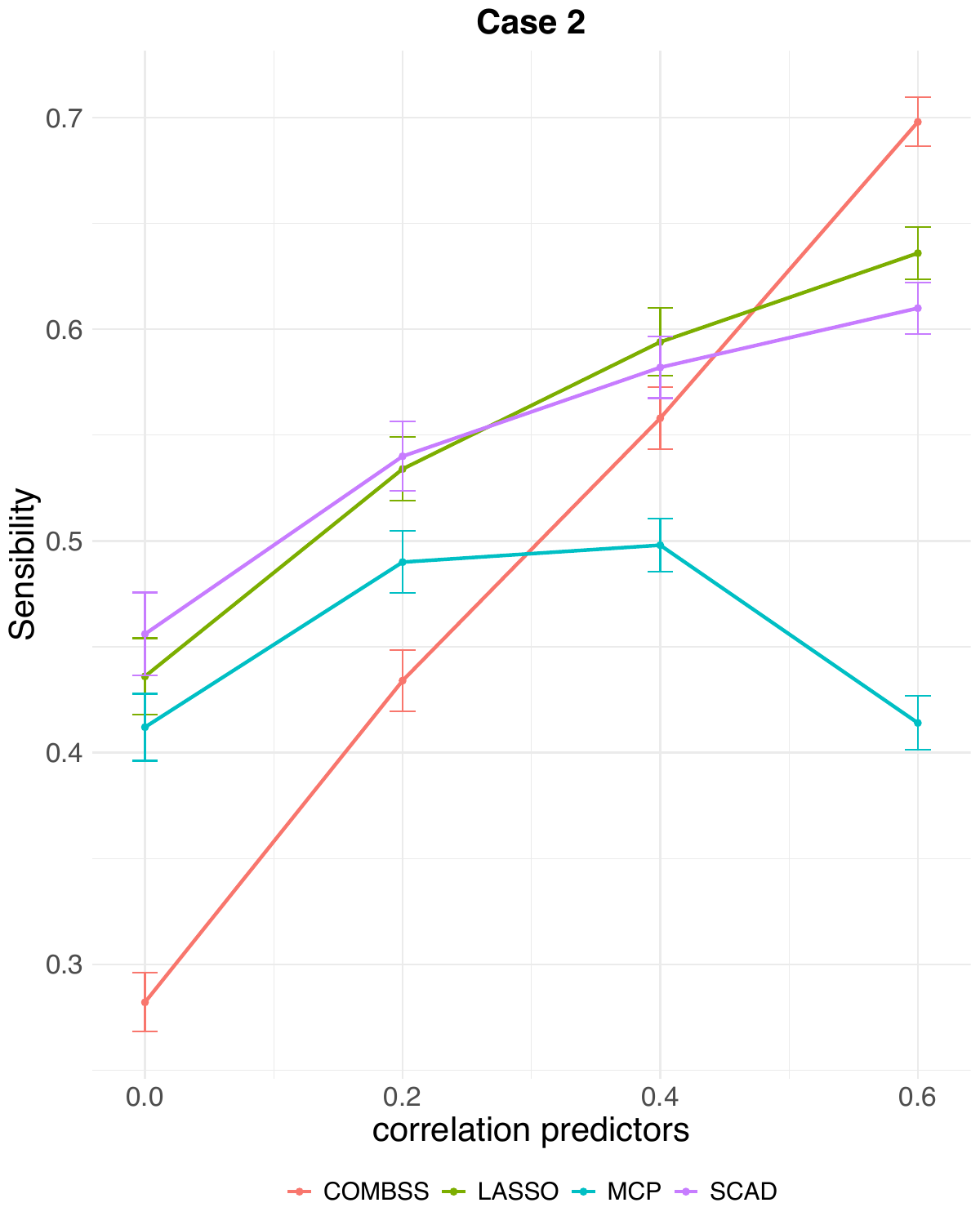}
    \caption{Sensitivity -- Case~2}
    \label{fig:high_case2_sens}
  \end{subfigure}
  \hfill
  \begin{subfigure}[b]{0.32\textwidth}
    \centering
    \includegraphics[width=\textwidth]{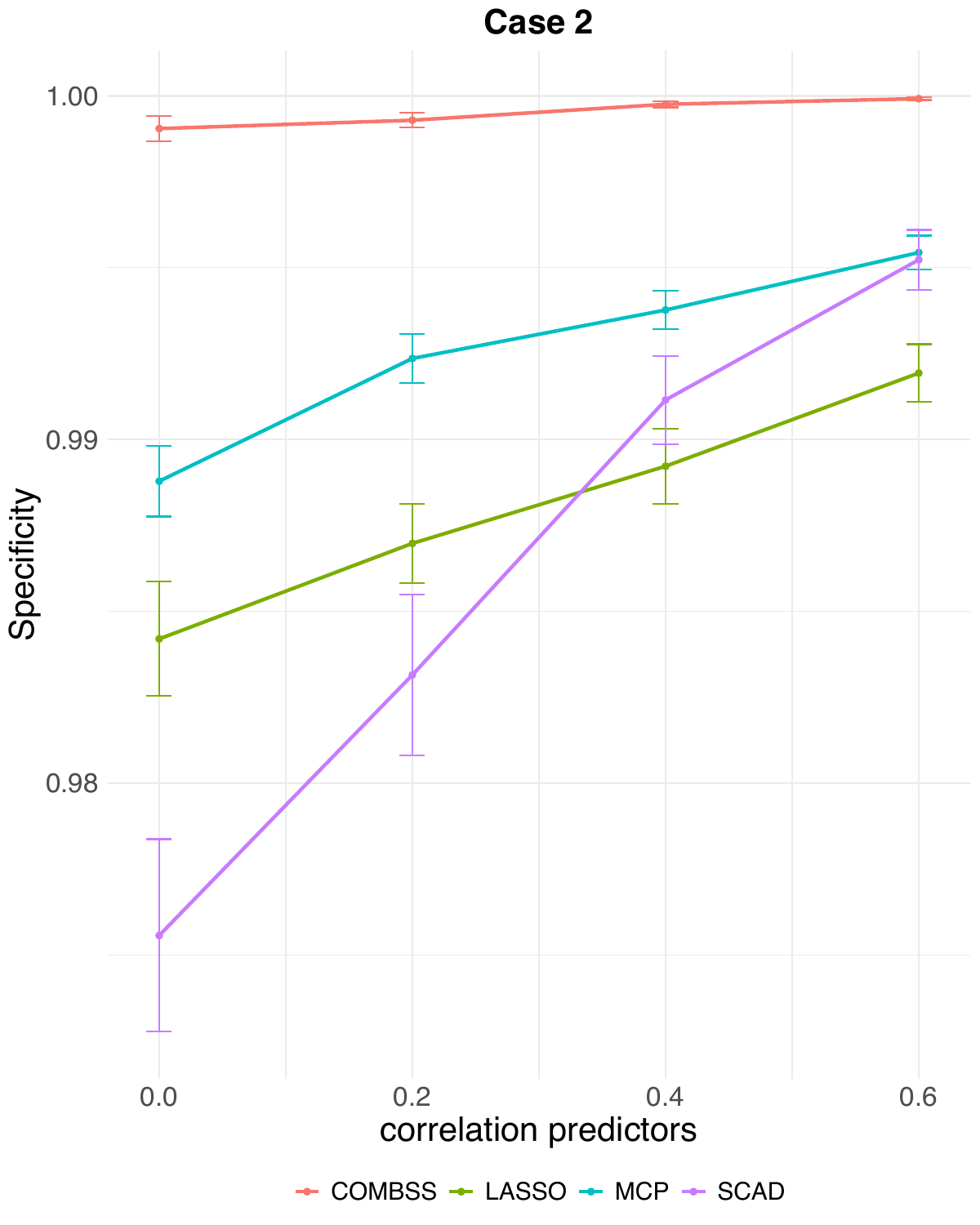}
    \caption{Specificity -- Case~2}
    \label{fig:high_case2_spec}
  \end{subfigure}
  \hfill
  \begin{subfigure}[b]{0.32\textwidth}
    \centering
    \includegraphics[width=\textwidth]{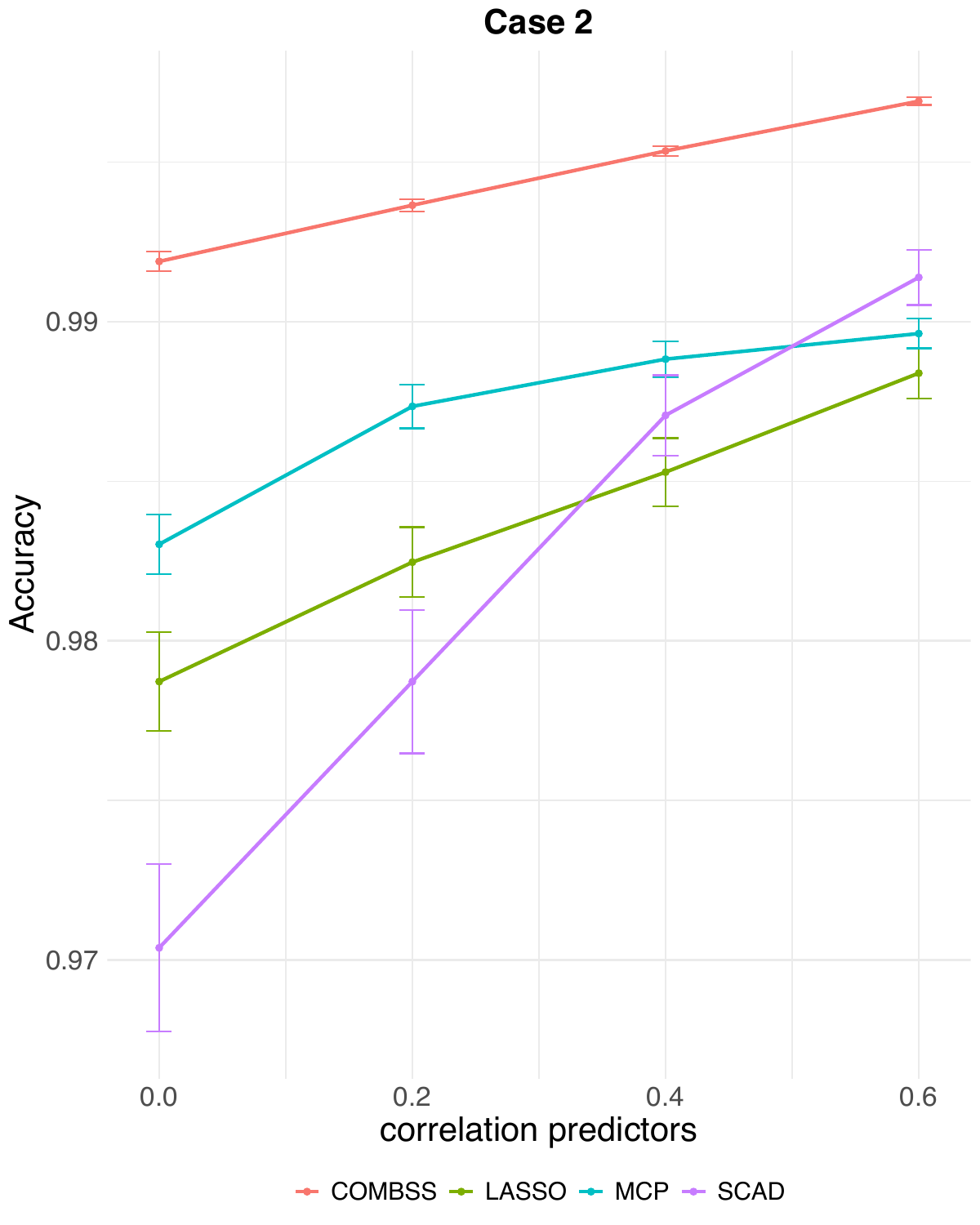}
    \caption{Selection accuracy -- Case~2}
    \label{fig:high_case2_selacc}
  \end{subfigure}
  \caption{Supplementary performance results in the high-dimensional setting
    ($n = 200$, $p = 1000$) for Case~1 (top row) and Case~2 (bottom row).
    Each panel displays the average sensitivity, specificity, and selection
    accuracy over 50 replications as a function of the predictor correlation
    $\rho \in \{0, 0.2, 0.4, 0.6\}$, with vertical bars denoting one
    standard error.}
  \label{fig:supp_sim_high}
\end{figure}

\clearpage
\section{GWAS Application: Extended Results}
\label{app:gwas-extended}

\begin{figure}[H]
  \centering
  \begin{subfigure}[b]{0.48\textwidth}
    \centering
    \includegraphics[height=5cm]{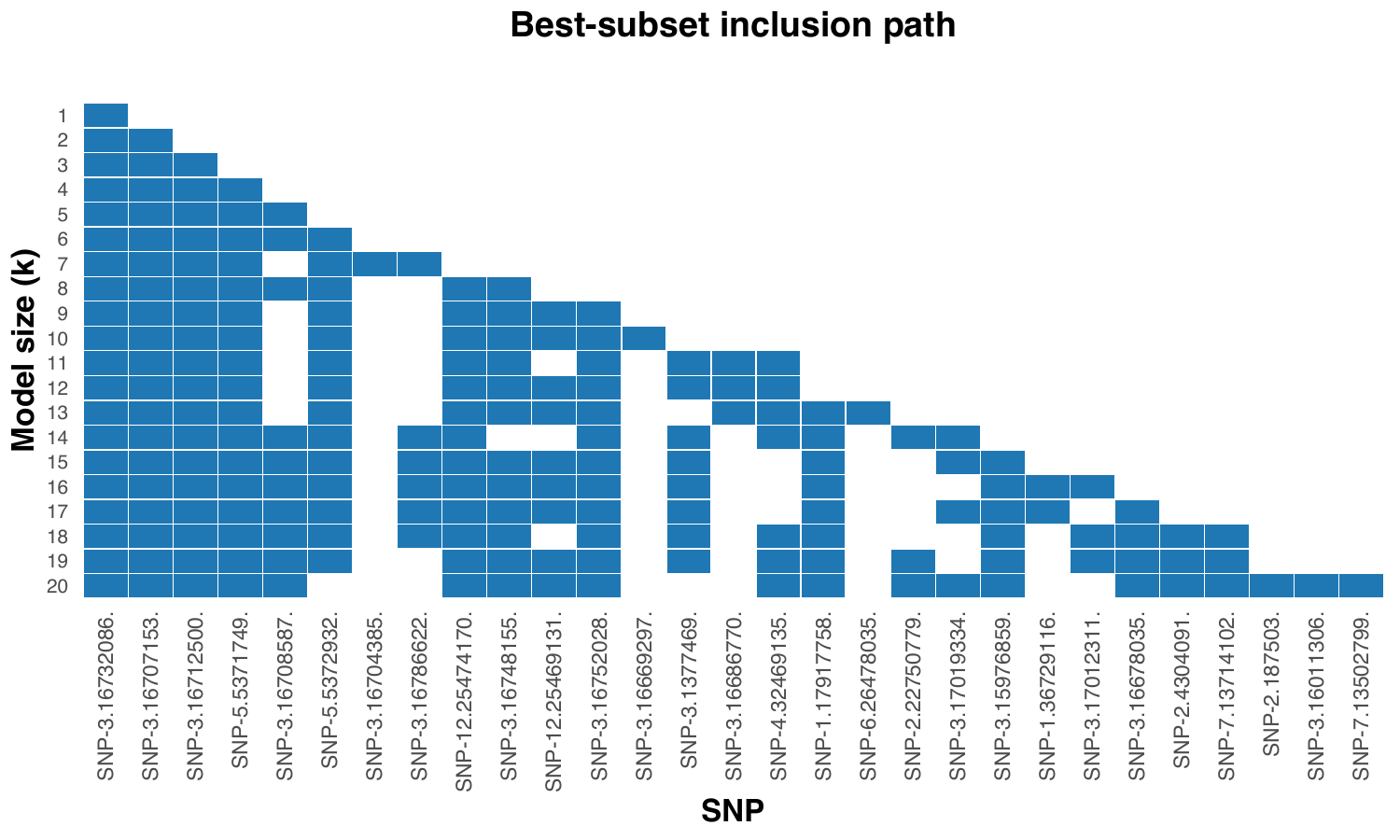}
    \caption{Best-subset inclusion path for models.}
    \label{fig:best_subset_inclusion_path_20}
  \end{subfigure}
  \hfill
  \begin{subfigure}[b]{0.48\textwidth}
    \centering
    \includegraphics[height=5cm]{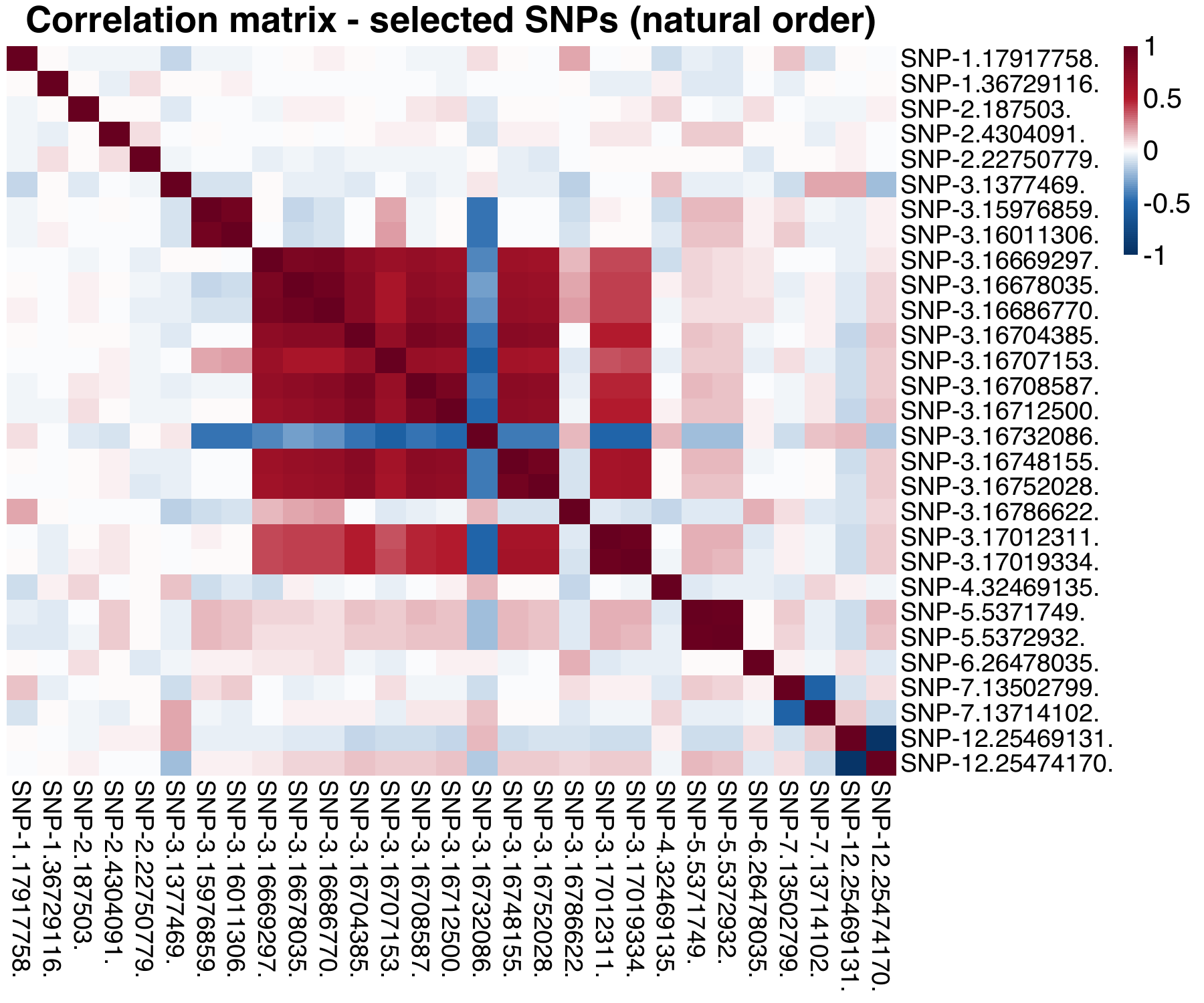}
    \caption{Correlation matrix of the selected SNPs.}
    \label{fig:correlation_matrix_20}
  \end{subfigure}
  \caption{(a) Best-subset inclusion path showing the selected SNPs for model
    sizes $k = 1$ to $k = 20$.  (b) Correlation matrix for the selected SNPs
    showcasing relationships between the predictors.}
  \label{supplefig:combined_figure}
\end{figure}

\begin{figure}[H]
  \centering
  \includegraphics[height=5cm]{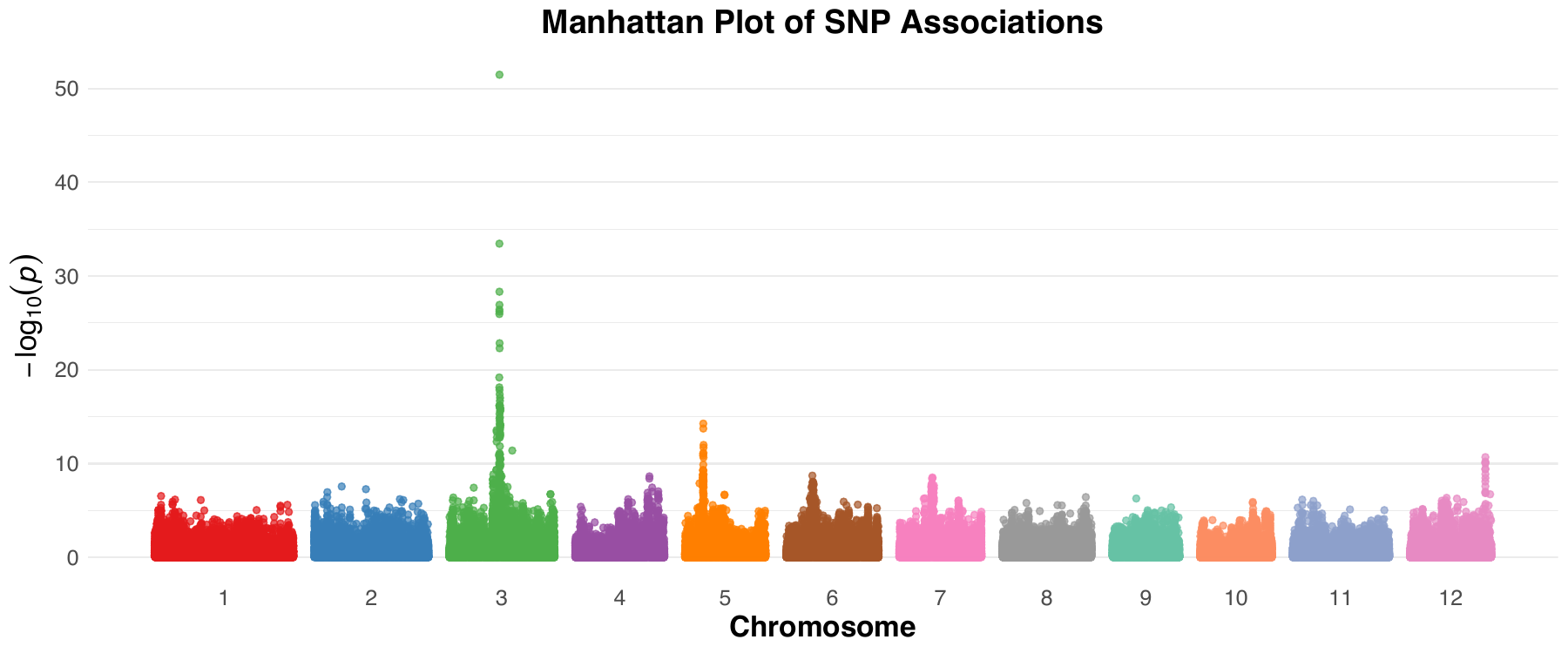}
  \caption{Manhattan plot based on unadjusted p-values using univariate student test}
  \label{manhattanplot}
\end{figure}

\clearpage
\section{Gene Annotations for the Khan SRBCT Application}
\label{app:khan_genes}

\begin{table}[H]
\centering
\caption{Annotation of the 16 genes selected by COMBSS-GLM at model size $k = 12$, which achieves perfect classification accuracy on the Khan SRBCT test set. The column ``Entry $k$'' indicates the smallest model size at which each gene first appears in the inclusion path. Gene descriptions and Image IDs are obtained from the \texttt{srbct} dataset in the \textsf{R} package \texttt{mixOmics} \citep{rohart2017mixomics}.} 
\label{tab:khan_genes}
\begingroup\small
\begin{tabular}{rrll}
  \toprule
Entry $k$ & Gene & Image ID & Gene Description \\ 
  \midrule
1 & 1954 & 814260 & follicular lymphoma variant translocation 1 \\ 
  2 & 1955 & 784224 & fibroblast growth factor receptor 4 \\ 
  3 & 246 & 377461 & caveolin 1, caveolae protein, 22kD \\ 
  4 & 1389 & 770394 & Fc fragment of IgG, receptor, transporter, alpha \\ 
  5 & 1645 & 52076 & olfactomedinrelated ER localized protein \\ 
  6 & 187 & 296448 & insulin-like growth factor 2 (somatomedin A) \\ 
  7 & 509 & 207274 & Human DNA for insulin-like growth factor II (IGF-2); exon 7 an... \\ 
  8 & 2050 & 295985 & ESTs \\ 
  9 & 545 & 1435862 & antigen identified by monoclonal antibodies 12E7, F21 and O13 \\ 
  10 & 1319 & 866702 & protein tyrosine phosphatase, non-receptor type 13 (APO-1/CD95... \\ 
  11 & 910 & 839552 & nuclear receptor coactivator 1 \\ 
  12 & 1074 & 1471841 & ATPase, Na+/K+ transporting, alpha 1 polypeptide \\ 
   \bottomrule
\end{tabular}
\endgroup
\end{table}

\clearpage
\section{Hyperparameter Tuning}
\label{sec:penalty-calibration}

The concavity threshold $\delta_{\mathrm{conc}}$ from
Theorem~\ref{prop:concavity-threshold-unified} provides a natural anchor
for the geometric penalty schedule. For logistic and multinomial regression, we can compute it using $\nu_{\max}$ as shown in \eqref{eqn:delta_conc}.
The eigenvalue $\nu_{\max}$ is not available in closed form for general $X_u$, but
can be computed efficiently by power iteration: starting from a random
unit vector $v$, one repeatedly forms $v \leftarrow X_u^\top X_u v /
\|X_u^\top X_u v\|_2$ and estimates $\nu_{\max} \approx \|X_u v\|_2^2$.
Convergence is geometric and requires only matrix--vector products with
$X_u$ and $X_u^\top$, making the cost negligible relative to the inner
solves.

Given $\delta_{\mathrm{conc}}$, in our simulations, we set
\begin{equation*}
  \delta_{\max} = \delta_{\mathrm{conc}},\qquad\delta_{\min} = 10^{-3}\,\delta_{\mathrm{conc}},
\end{equation*}
so that the schedule begins well below the concavity threshold.

\paragraph{Controlling complexity via $N$}
The iteration budget $N$ is the primary user-facing parameter governing
computational cost: by Remark~\ref{rem:complexity}, the total cost of
Algorithm~\ref{alg:fw-homotopy-joint} is exactly $N$ calls to the inner
GLM solver, so the user can trade accuracy against runtime by adjusting
$N$ alone, without changing any other hyperparameter.
Importantly, the geometric penalty schedule traverses the full range
$[\delta_{\min},\delta_{\max}]$ for any $N\ge 2$, with growth rate
$r=(\delta_{\max}/\delta_{\min})^{1/N}$; a larger $N$ takes finer steps
along this path but reaches the same endpoint $\delta_{\max}=
\delta_{\mathrm{conc}}$.
Consequently, the subset selected by the algorithm is not sensitive to
$N$ once $N$ is moderately large: the concave landscape at
$\delta_{\max}$ is identical regardless of $N$, and the Frank--Wolfe
iterates are guided toward the same binary corner.

\begin{table}[H]
\centering
\caption{$N$-sensitivity analysis (Case~1, $p=30$, $\rho=0.6$, 10 replications).}
\label{tab:N-sensitivity}
\begin{tabular}{c cc cc}
\hline
$N$ & \multicolumn{2}{c}{$k$} & \multicolumn{2}{c}{Test accuracy}\\
\cline{2-3} \cline{4-5}
 & Logistic & Multinomial & Logistic & Multinomial\\
\hline
$50$   & $9.9$  & $13.3$ & $0.89$ & $0.68$\\
$100$  & $9.9$  & $11.2$ & $0.89$ & $0.69$\\
$500$  & $9.9$  & $8.6$  & $0.89$ & $0.69$\\
$1000$ & $10.6$ & $8.4$  & $0.88$ & $0.69$\\
\hline
\end{tabular}
\end{table}


\end{document}